\documentclass[11pt]{article}
\usepackage{amsmath}
\usepackage{amssymb}

\usepackage{scicite}

\usepackage{subfigure}
\usepackage{times}
\usepackage{epsfig}
\usepackage{url}

\newtheorem{definition}{Definition}[section]
\newtheorem{lemma}{Lemma}[section]
\newtheorem{theorem}{Theorem}[section]

\newtheorem{proof}{Proof}
\newcommand{\Var}{\textrm{Var}}

\newenvironment{sciabstract}{%
\begin{quote} \bf}
{\end{quote}}

\newcounter{lastnote}

\begin{document}


\title{ Provable Security of Networks
\footnote{State Key Laboratory of Computer Science, Institute of Software,
Chinese Academy of Sciences, P. O. Box 8718, Beijing, 100190, P. R.
China.  Email: \{angsheng, yicheng, zhangw\}@ios.ac.cn.
Correspondence: \{angsheng, yicheng\}@ios.ac.cn.
\newline
Angsheng Li is partially supported by the Hundred-Talent Program of
the Chinese Academy of Sciences. All authors are supported by the
Grand Project ``Network Algorithms and Digital Information" of the
Institute of software, Chinese Academy of Sciences, and NSFC grant
No. 61161130530.}}

\author{ Angsheng Li$^{1}$,  Yicheng Pan$^{1,3 }$, Wei Zhang$^{1,2}$ \\
\normalsize{$^{1}$State Key Laboratory of Computer Science}\\
\normalsize{ Institute of Software, Chinese Academy of Sciences}\\
\normalsize{$^{2}$University of Chinese Academy of Sciences,
P. R. China}\\
\normalsize{$^{3}$State Key Laboratory of Information Security}\\
\normalsize{ Institute of Information Engineering, Chinese Academy of Sciences,
P. R. China} }


\date{}



\baselineskip24pt

\maketitle

\begin{sciabstract}

We propose a {\it security hypothesis} that a network is {\it
secure}, if any deliberate attacks of a small number of nodes will
never generate a global failure of the network, and a {\it
robustness hypothesis} that a network is {\it robust}, if a small
number of random errors will never generate a global failure of the
network. Based on these hypotheses, we propose a definition of {\it
security} and a definition of {\it robustness} of networks against
the cascading failure models of deliberate attacks and random errors
respectively, and investigate the principles of the security and
robustness of networks. We propose a {\it security model} such that
networks constructed by the model are provably secure against any
attacks of small sizes under the cascading failure models, and
simultaneously follow a power law, and have the small world property
with a navigating algorithm of time complex $O(\log n)$. It is shown
that for any network $G$ constructed from the security model, $G$
satisfies some remarkable topological properties, including: (i) the
{\it small community phenomenon}, that is, $G$ is rich in
communities of the form $X$ of size poly logarithmic in $\log n$
with conductance bounded by $O(\frac{1}{|X|^{\beta}})$ for some
constant $\beta$, (ii) small diameter property, with diameter
$O(\log n)$ allowing a navigation by a $O(\log n)$ time algorithm to
find a path for arbitrarily given two nodes, and (iii) power law
distribution, and satisfies some probabilistic and combinatorial
principles, including the {\it degree priority theorem}, and {\it
infection-inclusion theorem}. These properties allow us to prove
that almost all communities of $G$ are strong, where a community is
strong if the seed (or hub) of the community cannot be infected by
the collection of its neighbor communities unless some node of the
community itself is targeted or has already been infected, and more
importantly that there exists an {\it infection priority tree $T$ of
$G$} such that infections of a strong community must be triggered by
an edge in the infection priority tree $T$, and such that the
infection priority tree $T$ has height $O(\log n)$. By using these
principles, we show that a network $G$ constructed from the security
model is secure for any attacks of small scales under both the
uniform threshold and random threshold cascading failure models. Our
security theorems show that networks constructed from the security
model are provably secure against any attacks of small sizes, for
which natural selections of {\it homophyly, randomness} and {\it
preferential attachment} are the underlying mechanisms. We also show
that networks generated from the preferential attachment (PA, for
short) model satisfy a {\it threshold theorem of robustness} of
networks with a constant threshold so that the networks constructed
from the PA model cannot be even robust against random errors of
small sizes under the uniform threshold cascading failure model. We
design and implement an experiment which shows that overlapping
communities undermine security of networks.  Our results here
explore that security of networks can be achieved theoretically by
structure of networks, that there is a tradeoff between the role of
structure and the role of thresholds in security of networks, and
that neither power law nor small world property is an obstacle of
security of networks. The proofs of our results provide a general
framework to analyze security of networks.

\end{sciabstract}

Network security has been a fundamental issue from the very
beginning of network science due to its great importance to all the
applications of networks such as the internet, social science,
biological science, and economics etc. In the last few years,
security of networks has become an urgent challenge in network
applications.

Clearly, security depends on attacks of networks. Typical attacks
include both {\it physical attack} of removal of nodes or edges and
{\it cascading failure models of attacks}, similar to that of
viruses spreading. In the case of  physical attacks of removal of
nodes to destroy the global connectivity of networks, it was
shown~\cite{AJB2000} that many networks, including the
world-wide-web, the internet, social networks, are extremely
vulnerable to intentional attacks of removal of a small fraction of
high degree nodes, but at the same time display a high degree of
robustness against random errors.

The second type of attacks is the cascading failure model, see for
instance~\cite{AM91}, \cite{M00}, \cite{W02}, \cite{SFSVVD2009}.
This model captures the behaviors of spreading of information, of
viruses on computer networks, of news on internet, of ideas on
social networks, and of influence in economic networks etc. There
are different definitions of diffusions in networks in the
literature. Here we investigate the \emph{threshold cascading
failure model} which was formulated in social studies, and used in
simulating the epidemic spread in networks \cite{G78}. In this
model, the members have a binary decision and are influenced by
their neighbors in scenarios such as rumor spreading, disease
spreading, voting, and advertising etc. This model of cascading
behavior has been studied in physics, sociology, biology, and
economics~\cite{N03}, ~\cite{W02},~\cite{AM91}, ~\cite{M00}.

Blume et al. studied the algorithmic aspect of the threshold
cascading failure model on regular graphs of different patterns,
particularly on cliques and trees~\cite{BEKKT11}. Kempe et al.
considered the influence maximization problem for the linear
threshold model and gave a $(1-\frac{1}{e})$-approximation algorithm
based on the sub-modularity of influence functions~\cite{KKT03}.

In the present paper, we propose a theory of security of complex
networks. First of all, we need to understand what exactly factors
of networks determine the security of the networks. We found that
security of a network, $G$ say, depends on the following objects:

\begin{itemize}
\item \ Strategies of attacks
\item \ Topological structure of the network
\item \ Probabilistic principles
\item \ Combinatorial principles
\item \ The sizes of attacks
\item \ The cost of failures
\item \ Thresholds of vertices, for cascading failure models

\end{itemize}

A theory is to investigate the mathematical relationships among
these objects.

\section{Security and Robustness Hypotheses}\label{sec:inf and inj}

In this section, we introduce the basic definitions for us to
quantitatively analyze the security and robustness of networks.

We define the threshold cascading failure model as follows.

\begin{definition}\label{def:cascading} (Infection set) Let
$G=(V,E)$ be a network. Suppose that for each node $v\in V$, there
is a threshold $\phi (v)$ associated with it. For an initial set
$S\subset V$, the {\it infection set} of $S$ in $G$ is defined
recursively as follows:

\begin{enumerate}
\item [(1)] Each node $x\in S$ is called {\it infected}.
\item [(2)] A node $x\in V$ becomes infected, if it has not been infected yet,
and $\phi (x)$ fraction of its neighbors have been infected.

\end{enumerate}

We use ${\rm inf}_G(S)$ to denote the infection set of $S$ in
$G$.

\end{definition}

The cascading failure models depend on the choices of thresholds
$\phi (v)$ for all $v$. We consider two natural choices of the
thresholds. The first is random threshold cascading, and the second
is uniform threshold cascading.

\begin{definition}\label{def:random threshold} (Random threshold)
We say that a cascading failure model is {\it random}, if for each
node $v$, $\phi (v)$ is defined randomly and uniformly, that is,
$\phi (v)=r/d$, where $d$ is the degree of $v$ in $G$, and $r$ is
chosen randomly and uniformly from $\{1,2,\cdots, d\}$.

\end{definition}

\begin{definition}\label{def:uniform threshold} (Uniform threshold)
We say that a cascading failure model is {\it uniform}, if for each
node $v$, $\phi (v)=\phi$ for some fixed number $\phi$.

\end{definition}

To compare the two strategies of physical attacks and cascading
failure models of attacks, we introduce the notion of {\it injury
set} of physical attacks.

\begin{definition}\label{def:injury set} (Injury set)
Let $G=(V,E)$ be a network, and $S$ be a subset of $V$. The physical
attacks on $S$ is to delete all nodes in $S$ from $G$. We say that a
node $v$ is injured by the physical attacks on $S$, if $v$ is not
connected to the largest connected component of the graph obtained
from $G$ by deleting all nodes in $S$.

We use ${\rm inj}_G(S)$ to denote the injury set of $S$ in $G$.

\end{definition}

In~\cite{LZPL2013b}, it was shown that cascading failure models of
attacks are better than that of physical attacks, by simulating the
attacks on networks of classical models of networks.

The first model is the Erd\"os-R\'enyi (ER, for short)
model~\cite{ER1959}, \cite{ER1960}. In this model, we construct
graph as follows: Given $n$ nodes $1,2,\cdots, n$, and a number $p$,
for any pair $i, j$ of nodes $i$ and $j$, we create an edge $(i,j)$
with probability $p$.

We depict the curves of sizes of the infection set and the injury
set of attacks of top degree nodes of networks of the ER model in
Figures ~\ref{fig:_cascading_vs_node_attack_ER_N=10000_d=10}
and~\ref{fig:_cascading_vs_node_attack_ER_N=10000_d=15}.

The second is the PA model~\cite{Bar1999}. In this model, we
construct a network by steps as follows: At step $0$, choose an
initial graph $G_0$. At step $t>0$, we create a new node, $v$ say,
and create $d$ edges from $v$ to nodes in $G_{t-1}$, chosen with
probability proportional to the degrees in $G_{t-1}$, where
$G_{t-1}$ is the graph constructed at the end of step $t-1$, and $d$
is a natural number.

We depict the comparisons of sizes of infection sets and injury sets
of attacks of the top degree nodes of networks generated from the
preferential attachment model in
Figures~\ref{fig:_cascading_vs_node_attack_PA_N=10000_d=10} and
~\ref{fig:_cascading_vs_node_attack_PA_N=10000_d=15}.

Figures ~\ref{fig:_cascading_vs_node_attack_ER_N=10000_d=10},
~\ref{fig:_cascading_vs_node_attack_ER_N=10000_d=15},
~\ref{fig:_cascading_vs_node_attack_PA_N=10000_d=10} and
~\ref{fig:_cascading_vs_node_attack_PA_N=10000_d=15} show that for
any network, $G$ say, generated from either the ER model or the PA
model, the following properties hold:

\begin{enumerate}
\item The infection sets are much larger than the corresponding
injury sets.

This means that to build our theory, we only need to consider the
attacks of cascading failure models.

\item The attacks of top degree nodes of size as small as $O(\log
n)$ may cause a constant fraction of nodes of the network to be
infected under the cascading failure models of attacks.

This means that networks of the ER and PA models are insecure for
attacks of sizes as small as $O(\log n)$.
\end{enumerate}

 Therefore
the main issue of network security is to resist the global cascading
failure of networks by attacks of sizes polynomial in $\log n$.

\begin{figure}
  \centering
  \subfigure[]
   {\label{fig:_cascading_vs_node_attack_ER_N=10000_d=10}
    \includegraphics[width=4in]{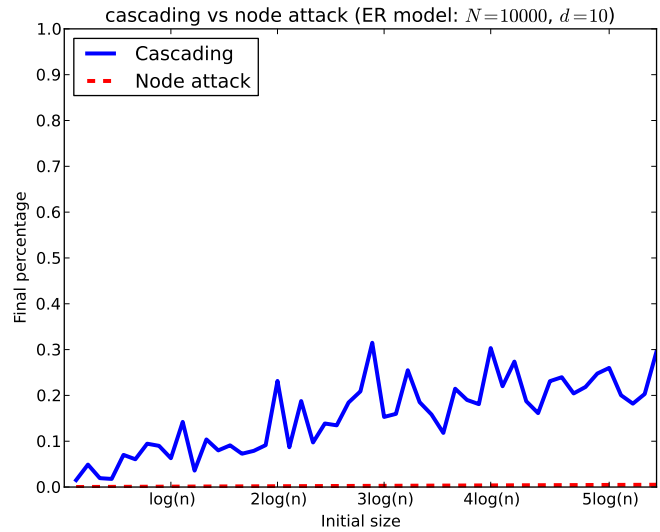}
   }

   \subfigure[]
   {\label{fig:_cascading_vs_node_attack_ER_N=10000_d=15}
    \includegraphics[width=4in]{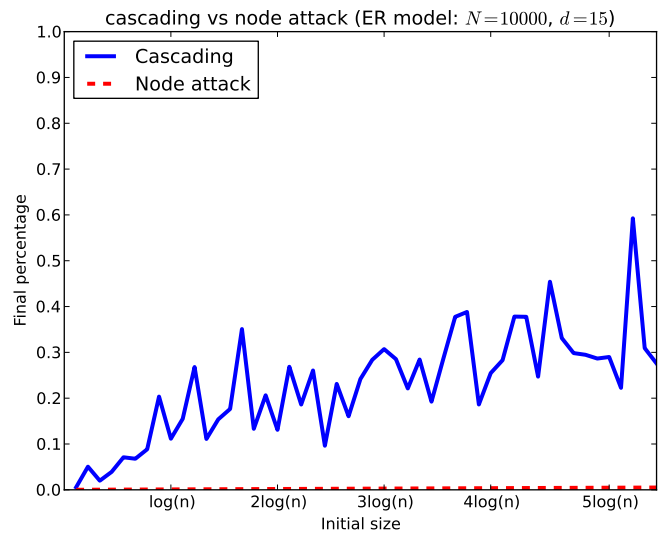}
   }

  \caption{\textbf{(a), (b) are the curves of fractions of sizes of infection sets and injury sets
  by attacks of the top degree nodes of small sizes, i.e., up to
  $5\log n$, for networks of the ER model for $n=10,000$ and for $d=10$ and $15$ respectively. The sizes of the infection sets are the
  largest ones among $100$ times attacks under random threshold cascading failure model. The infection sets and injury sets correspond
   to the blue and red curves respectively.}}
\end{figure}

\begin{figure}
  \centering
  \subfigure[]
   {\label{fig:_cascading_vs_node_attack_PA_N=10000_d=10}
    \includegraphics[width=4in]{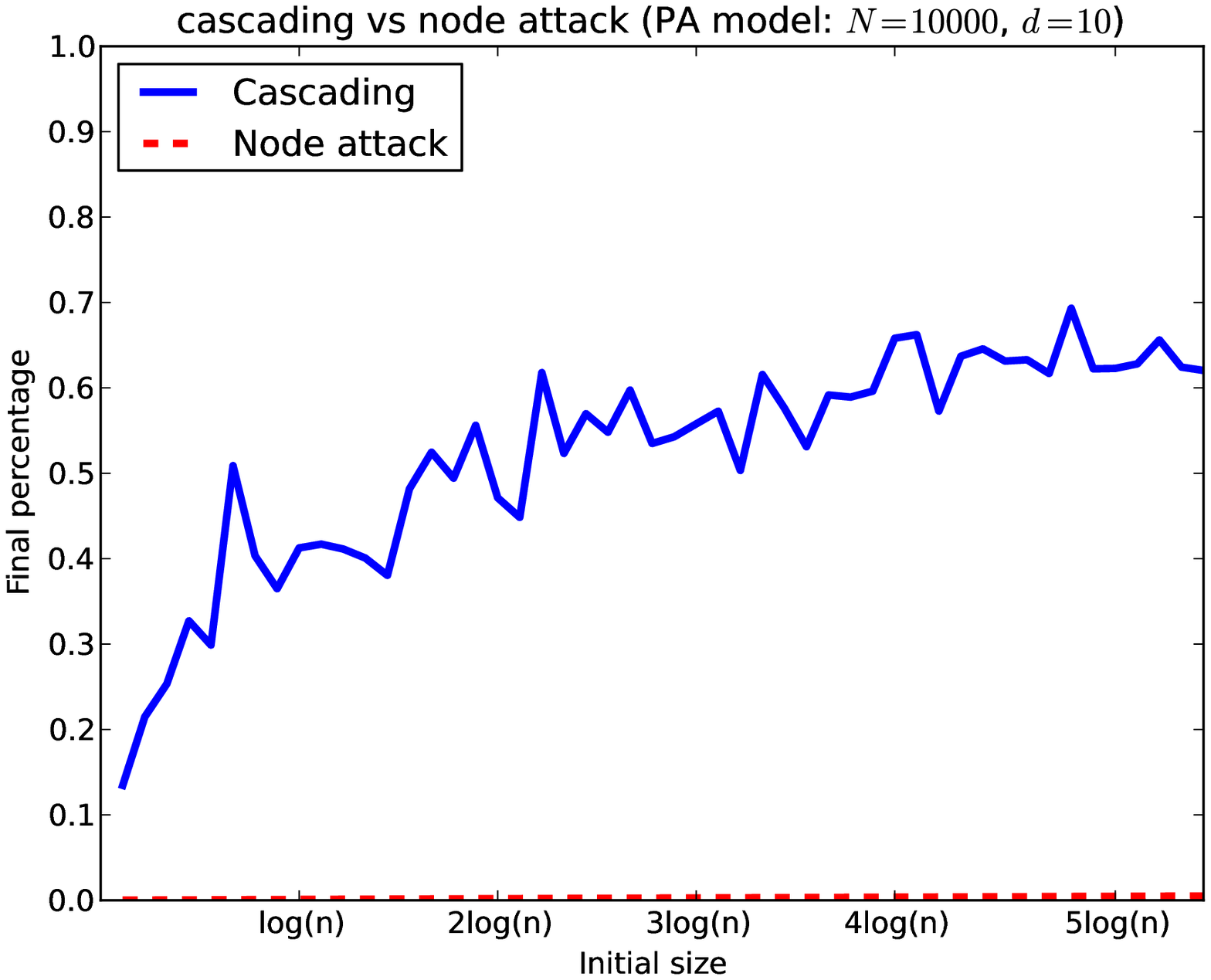}
   }

   \subfigure[]
   {\label{fig:_cascading_vs_node_attack_PA_N=10000_d=15}
    \includegraphics[width=4in]{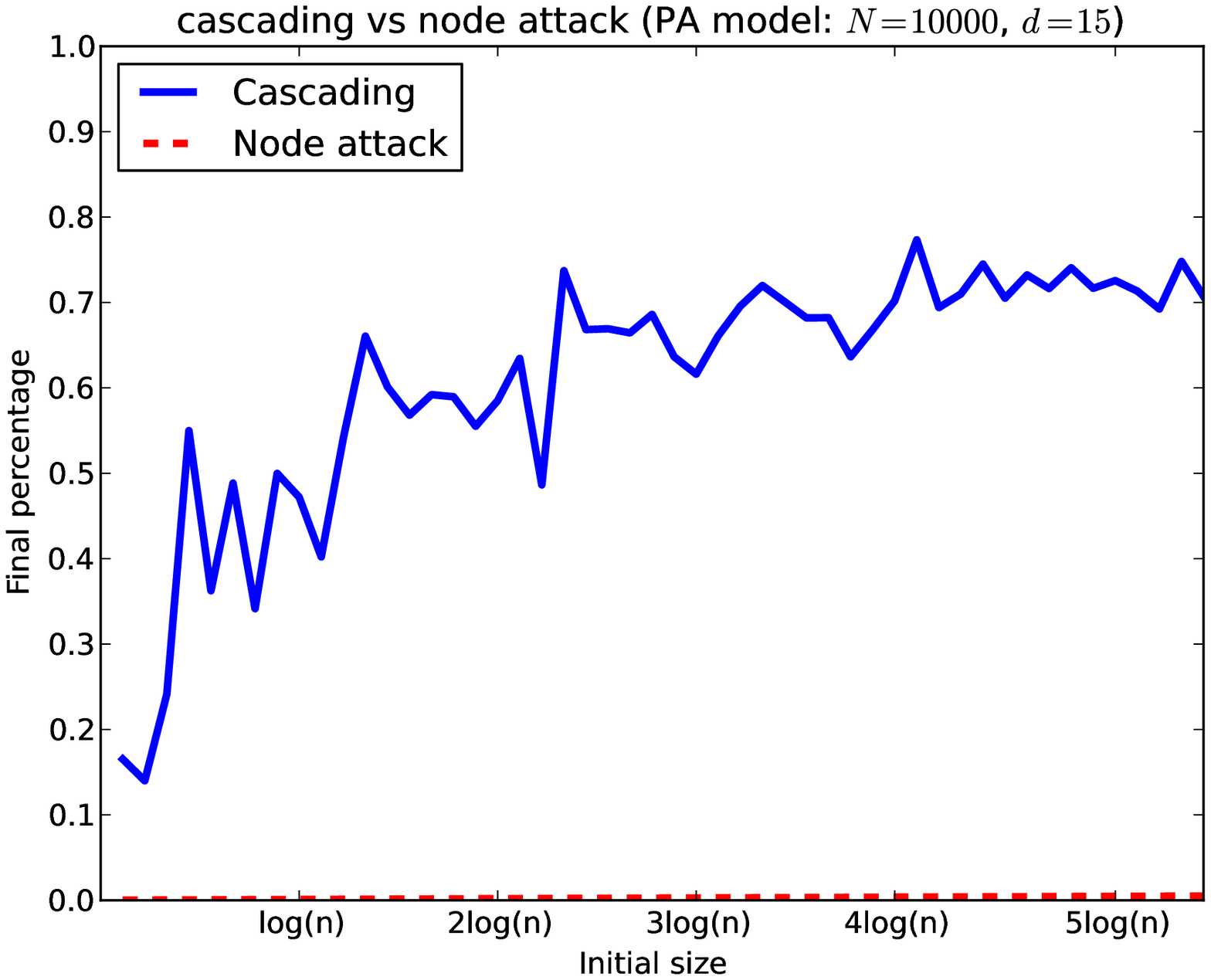}
   }

  \caption{\textbf{(a), (b) are the curves of fractions of sizes of infection sets and injury sets
  by attacks of the top degree nodes of small sizes, i.e., up to
  $5\log n$, for networks of the PA model for $n=10,000$ and for $d=10$ and $15$ respectively. The sizes of the infection sets are the
  largest ones among $100$ times attacks under random threshold cascading failure model. The infection sets and injury sets correspond
   to the blue and red curves respectively.}}
\end{figure}

From Figures ~\ref{fig:_cascading_vs_node_attack_ER_N=10000_d=10},
~\ref{fig:_cascading_vs_node_attack_ER_N=10000_d=15},
~\ref{fig:_cascading_vs_node_attack_PA_N=10000_d=10} and
~\ref{fig:_cascading_vs_node_attack_PA_N=10000_d=15}, we have that
the main issue of network security is to resist the global failure
of networks under cascading failure models, that for both theory and
applications, it suffices to guarantee the security against attacks
of sizes polynomial in $\log n$, and that topological structures of
networks are essential to the security of the networks, observed
from the comparison of infection fractions between the ER and the PA
models.

According to the experiments in Figures
~\ref{fig:_cascading_vs_node_attack_ER_N=10000_d=10},
~\ref{fig:_cascading_vs_node_attack_ER_N=10000_d=15},
~\ref{fig:_cascading_vs_node_attack_PA_N=10000_d=10} and
~\ref{fig:_cascading_vs_node_attack_PA_N=10000_d=15}, we propose the
following hypotheses.

{\bf Security Hypothesis}: We say that a network is {\it secure}, if
any small number of attacks of any strategy will never cause a
global failure of the network.

{\bf Robustness Hypothesis}: We say that a network is {\it robust},
if a small number of random errors of the network will never cause a
global failure of the network.

\section{Definitions of security and robustness}\label{def:security}

As mentioned in Section~\ref{sec:inf and inj}, the main issue is the
security for cascading failure models and for attacks of sizes
polynomial in $\log n$.

We propose mathematical definitions for security and robustness of
networks based on the security hypothesis and the robustness
hypothesis summarized in Section~\ref{sec:inf and inj},
respectively.

We consider the security of networks with arbitrary sizes. We define
the  security and robustness of networks under the threshold
cascading failure model as follows:

Let $n$ be the number of nodes of the network. We define

{\bf Security} With probability $1-o(1)$, the following event
occurs: For any initial set $S$ of size ${\rm poly}(\log n)$, $S$
will not cause a global cascading failure, that is, the size of the
infection set of $S$ in $G$ is $o(n)$.

and

 {\bf Robustness} With probability $1-o(1)$, a small number,
i.e., ${\rm poly}(\log n)$, of random choices of the initial set $S$
will not cause a {\it global cascading failure}, that is, the size
of the infection set of $S$ in $G$ is $o(n)$.

Let $\mathcal{M}$ be a model of networks. We investigate the
security of networks constructed from model $\mathcal{M}$. We define
the security of networks
 for attacks of cascading failure with both random threshold and uniform threshold
 respectively. Suppose that $G$ is a network of $n$
 nodes, constructed from model $\mathcal{M}$, for large $n$.

\begin{definition}\label{def:r-security} (Random threshold security)
For the cascading failure model of random threshold, we say that $G$
is {\it secure}, if almost surely, meaning that with probability
$1-o(1)$, the following holds:

for any set $S$ of size bounded by a polynomial of $\log n$, the
size of the infection set (or cascading failure set) of $S$ in $G$
is $o(n)$.

\end{definition}

\begin{definition}\label{def:u-security} (Uniform threshold
security) For the cascading failure model of uniform threshold, we
say that $G$ is {\it secure}, if almost surely, the following holds:
for an arbitrarily small $\phi$, i.e., $\phi =o(1)$, for any set $S$
of size bounded by a polynomial of $\log n$, $S$ will not cause a
global $\phi$-cascading failure, that is, the size of the infection
set of $S$ in $G$, written by ${\rm inf}^{\phi}_G(S)$, is bounded by
$o(n)$.

\end{definition}

\begin{definition}\label{def:security} (Security of model $\mathcal{M}$) Let $\mathcal{M}$ be a model
of networks. We say that model $\mathcal{M}$ is secure, if networks
constructed from model $\mathcal{M}$ are secure for both random and
uniform threshold cascading failure models of attacks.

\end{definition}

\begin{definition}\label{def:r-robustness} (Random threshold robustness)
For the cascading failure model of random threshold, we say that $G$
is {\it robust}, if almost surely, meaning that with probability
$1-o(1)$, the following holds:

for randomly chosen set $S$ of size bounded by a polynomial of $\log
n$, the size of the infection set of $S$ in $G$ is $o(n)$.

\end{definition}

\begin{definition}\label{def:u-robustness} (Uniform threshold
robustness) For the cascading failure model of uniform threshold, we
say that $G$ is {\it robust}, if almost surely, the following holds:
for an arbitrarily small $\phi$, i.e., $\phi =o(1)$, for randomly
chosen set $S$ of size bounded by a polynomial of $\log n$, $S$ will
not cause a global $\phi$-cascading failure, that is, the size of
the infection set of $S$ in $G$, written by ${\rm inf}_G^{\phi}(S)$,
is bounded by $o(n)$.

\end{definition}

\begin{definition}\label{def:robustness} (Robustness of model $\mathcal{M}$) Let $\mathcal{M}$ be a model
of networks. We say that model $\mathcal{M}$ is robust, if networks
constructed from model $\mathcal{M}$ are robust for both random and
uniform threshold cascading failure models of random errors.

\end{definition}

In Definitions~\ref{def:r-security}, ~\ref{def:u-security},
~\ref{def:r-robustness} and ~\ref{def:u-robustness}, the sizes of
attacks or random errors are polynomial in $\log n$. This is
sufficient for both theory and applications. The reason is that
networks constructed from both the ER and PA models are insecure, in
the sense that attacks of $O(\log n)$ top degree nodes may generate
a constant fraction of nodes of the networks to be infected, as
shown in Figures
~\ref{fig:_cascading_vs_node_attack_ER_N=10000_d=10},
~\ref{fig:_cascading_vs_node_attack_ER_N=10000_d=15},
~\ref{fig:_cascading_vs_node_attack_PA_N=10000_d=10} and
~\ref{fig:_cascading_vs_node_attack_PA_N=10000_d=15}.

\section{Security model of networks: algorithms and principles}\label{sec:sec-model}

From Figures ~\ref{fig:_cascading_vs_node_attack_ER_N=10000_d=10},
~\ref{fig:_cascading_vs_node_attack_ER_N=10000_d=15},
~\ref{fig:_cascading_vs_node_attack_PA_N=10000_d=10} and
~\ref{fig:_cascading_vs_node_attack_PA_N=10000_d=15}, we know that
nontrivial networks of both the ER model and the PA model are
insecure. This poses fundamental questions such as: Are there
networks with power law and small world property that are secure by
Definitions~\ref{def:r-security} and \ref{def:u-security}? What
mechanisms guarantee the security of networks? Is there any
algorithm to construct secure networks?

In~\cite{LZPL2013a}, the authors proposed a security model of
networks, and showed by experiments that networks of the security
model are much more secure than that constructed from both the ER
and PA models.

\begin{definition} \label{def:Securitymodel} (Security model)
Let $d\geq 4$ be a natural number and $a$ be a real number, which is
called {\it homophyly exponent}. We construct a network by stages.

\begin{enumerate}
\item Let $G_2$ be an initial  graph such that each node is associated with a
distinct {\it color}, and called {\it seed}.

\item Let $i>2$. Suppose that $G_{i-1}$ has been defined.
Define  $p_i=(\log i)^{-a}$.

\item With probability $p_i$, $v$ chooses a new color, $c$ say. In this
case, do:

\begin{enumerate}

\item we say that
$v$ is the {\it seed node} of color $c$,

\item (Preferential attachment scheme) add an edge $(u,v)$, such
 that $u$ is chosen with probability proportional to the degrees of nodes
in $G_{i-1}$, and

\item (Randomness) add $d-1$ edges $(v,u_j)$, $j=1,2,\ldots,d-1$, where $u_j$'s are
chosen randomly and uniformly among all seed nodes in $G_{i-1}$.
\footnote{If all the newly created $d$ edges linking from $v$ to
nodes in $G_{i-1}$ are chosen with probability proportional to their
degrees, then the model is the homophyly model \cite{LLPP2012a}.}
\end{enumerate}

\item (Homophyly and preferential attachment) Otherwise. Then $v$ chooses an old color, in which
case, then:

\begin{enumerate}
\item let $c$ be a color chosen randomly and uniformly among all colors in
$G_{i-1}$,

\item define the color of $v$ to be $c$ , and

\item add $d$ edges $(v,u_j)$, for $j=1,2,\ldots,d$, where $u_j$'s are
chosen with probability proportional to the degrees of all the nodes
that have the same color as $v$ in $G_{i-1}$.

\end{enumerate}
\end{enumerate}
\end{definition}

 It is clear that
Definition~\ref{def:Securitymodel} is a dynamic model of networks
for which homophyly, randomness and preferential attachment are the
underlying mechanisms.

As shown in~\cite{LZPL2013a}, ~\cite{LZPL2013b},  networks
constructed from the security model are much more secure than that
of the ER and PA models. To understand the intuition of the security
model, we use a figure in~\cite{LZPL2013b},
Figure~\ref{fig:_cascading_3models_N=10000_d=10_a=1.5} here. It
depicts three curves of sizes of infection sets of attacks of top
degree nodes of sizes up to $5\cdot\log n$ under random threshold
cascading failure model on networks generated from the security
model, the ER model and the PA model respectively. The curves
correspond to the largest infection set among $100$ times of attacks
over random choices of thresholds of the networks. The figure shows
that networks of the security model are in deed much more secure
than that of both the ER and the PA models, even if we just take the
homophyly exponent $a>1$ in the security model.

\begin{figure}
  \centering
\includegraphics[width=4in]{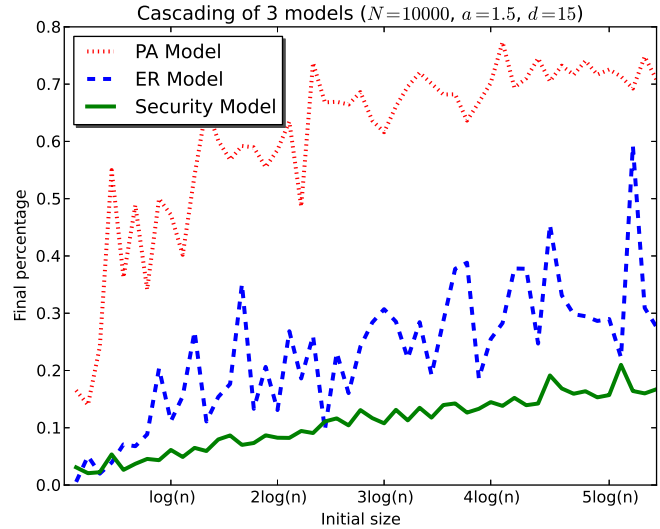}
\caption{\textbf{The curves are cascading failures of networks of
the ER model, the PA model and the security model for $n=10,000$,
$d=15$ and
$a=1.5$}}\label{fig:_cascading_3models_N=10000_d=10_a=1.5}

\end{figure}

Experiments in~\cite{LZPL2013a} showed the following properties:

\bigskip

\begin{enumerate}
 \item The mechanisms of homophyly, randomness
and preferential attachment ensure that networks of the security
model satisfy a number of structural properties such as:

\begin{enumerate}
\item (Small community phenomenon) A network, $G$ say, is rich in quality communities of small
sizes.

In fact, let $S$ be a homochromatic set of $G$. Then the induced
subgraph of $S$, written by $G_S$, is highly connected, and the
conductance of $S$, written by $\Phi (S)$, is bounded by a number
reversely proportional to a constant power of the size of the
community, i.e., less than or equal to, $O(\frac{1}{|S|^{\beta}})$,
for some constant $\beta$, where $|S|$ is the size of $S$.

\item (Internal centrality) Each community is the
induced subgraph of nodes of the same color, which follows a
preferential attachment, and hence has only a few nodes dominating
the internal links of the community.

This shows a remarkable local heterogeneity of the networks.

\item (External centrality) Each community has a few nodes,
including the seed of the community, which dominate the external
links from the community to outside of the community.

\end{enumerate}

\item (Power law) The networks follow a power law.
\item (Small world property) The networks have small diameters.

\item (Global Randomness and uniformity) There is a high degree of
randomness and uniformity among the edges between nodes of different
colors.

This shows that the networks have a global homogeneity and a global
randomness.

\item A non-seed node, $x$ in a community $G_X$, created at time step $t$ can be
infected by a neighbor community $G_Y$, only if the seed node $y_0$
of $G_Y$ is created at a time step $s>t$ and an edge $(y_0,x)$ is
created by (3) (b) of Definition \ref{def:Securitymodel}.

\end{enumerate}

 The structural properties in (1) above allow us to develop a methodology
  of community analysis of networks. (2) and (3) show that the
  networks constructed from the security model have the most
  important properties of usual networks. (4) and (5) ensure that
  infections among different communities are hard. This intuitively
  explains the reason why networks constructed from the security
  model show much better security than that of the classic ER and PA
  models.

 The arguments above imply that the small community phenomenon, local heterogeneity, global
homogeneity and global randomness are essential to the security of
networks with power law and small world property.

In the present paper, we will show that the security model is
provably secure by Definition~\ref{def:security}. The key idea of
the proofs is a merging of some principles of topology, probability
and combinatorics.

We use $\mathcal{S}(n,a,d)$ to denote the set of random graphs of
$n$ nodes constructed by the security model with homophyly exponent
$a$ and average number of edges $d$ \footnote{In both Definition of
the PA model and the security model in~\ref{def:Securitymodel}, we
consider $d$ as a constant. Thus in all notations of
$O(\cdot),o(\cdot),\Omega(\cdot)$ and $\omega(\cdot)$ in the paper,
$d$ is always absorbed.}.

Let $G$ be a network constructed from the security model. We have
that each node is assigned a color. This new dimension of colors
allows us to characterize the structures of the networks. In our
security model, every node has its own characteristics from the very
beginning of its birth. This feature is remarkably different from
the classic models such as the ER and the PA models. Anyway, the
extra dimension of colors is essential to our understanding of
security of networks.

 We call a set of nodes of the same color, $\kappa$
say, a {\it homochromatic set}, written by $S_\kappa$.

We say that an edge is a {\it local edge} if two of its endpoints
share the same color, and {\it global edge}, otherwise.

At first, we prove some structural properties of networks of the
security model.

\bigskip

\begin{theorem} \label{thm:Securityproperties} (Fundamental theorem of the security
model) Let $a>1$ be the homophyly exponent, and $d\geq 4$ be a
natural number. Let $G=(V,E)$ be a network constructed by
$\mathcal{S}(n,a,d)$.

Then with probability $1-o(1)$, the following properties hold:

\begin{enumerate}

\item [(1)] (Basic properties):

\begin{enumerate}

\item [(i)] (Number of seed nodes is large) The number of seed nodes is bounded in the interval $[\frac{n}{2\log^a n},\frac{2n}{\log^a
n}]$.

\item [(ii)] (Communities whose vertices are interpretable by common features are small) Each homochromatic set
 has a size bounded by $O(\log^{a+1} n)$.

\end{enumerate}

\item [(2)] For degree distributions, we have:

\begin{itemize}

\item [(i)] (Internal centrality) The degrees of the induced
subgraph of a homochromatic set follow a power law.

\item [(ii)] The degrees of nodes of a homochromatic set follow a
power law.
\item [(iii)] (Power law) Degrees of nodes in $V$ follow a power law.
\end{itemize}

\item [(3)] For node-to-node distances, we have:
\begin{itemize}

\item [(i)] (Local communication law) The induced subgraph of a homochromatic set has a diameter bounded by $O(\log\log
n)$.
\item [(ii)] (Small world phenomenon) The average node to node distance of $G$ is bounded by $O(\log
n)$.
\item [(iii)] (Local algorithm to find short path between two nodes)
There is an algorithm to find a short path between arbitrarily given
two nodes in time $O(\log n)$.
\end{itemize}

\item [(4)] (Small community phenomenon) There are $1-o(1)$ fraction of nodes of $G$ each of which belongs to
a homochromatic set, $W$ say, such that the size of $W$ is bounded
by $O(\log^{a+1} n)$, and that the conductance of $W$, $\Phi (W)$,
is bounded by $O\left(\frac{1}{|W|^{\beta}}\right)$ for
$\beta=\frac{a-1}{4(a+1)}$.

This shows that the network is rich in quality communities of small
sizes.
\end{enumerate}
\end{theorem}

\smallskip

Theorem~\ref{thm:Securityproperties} explores an interesting
topology of a network $G$: (i) $G$ consists of a local structure and
a global structure, (ii) the local structure of $G$ is determined by
the small communities which have a number of local properties, and
(iii) the global structure of $G$ follows its own laws. The network
is rich in quality communities of small sizes which compose the
interpretable local structures of the network. On the other hand,
there is a global structure of the network which ensures that the
whole network is highly connected, with a power law distribution,
and a small diameter property. Communications in $G$ have two types,
the first is the local communications within the small communities
of length $O(\log\log n)$ and the second is the global ones which
make the whole network to be highly connected of length $O(\log n)$.
More importantly, there exists a {\it local algorithm} running in
time $O(\log n)$ to navigate in the whole network. Most of the
communications are local ones having length within $O(\log\log n)$,
and the rest of communications are global ones with length bounded
by $O(\log n)$. The construction of a network with explicit marks of
local and global structures by Definition~\ref{def:Securitymodel}
allows local algorithms of time complex  $O(\log n)$ to find useful
information in the whole network. This suggests a new algorithmic
problem, that is, to find network algorithms of time complexity
polynomial in $\log n$ for finding useful information.

Theorem~\ref{thm:Securityproperties} ensures that all the
communities are small. This guarantees that even if a single node in
a small community infects the whole community, the cascading failure
is still a local cost. However it is not intuitive to understand
from Theorem~\ref{thm:Securityproperties} the reason why networks of
the security model are secure. In fact, to prove the security
theorems, we need to develop some probabilistic and combinatorial
properties of the networks. In ~\cite{LZPL2013a}, the authors
analyzed experimentally some of these properties.

Suppose that $G=(V,E)$ is a network constructed from the security
model. For a subset $X\subset V$, we always use $G_X$ to denote the
induced subgraph of $X$ in $G$.

For a set of nodes $S$, we define $C(S)$ to be the set of colors
that appear in $S$. For a node $v$, we use $N(v)$ to denote the set
of neighbors of $v$. Given a node $v$, we define the {\it length of
degrees of $v$} to be the number of colors associated with the
neighbors of $v$, i.e., $|C(N(v))|$, written by $l(v)$.

Suppose that $N_1, N_2,\cdots, N_l$ are all the neighbors of $v$
such that nodes in each $N_i$ share the same color, and that nodes
in different $N_i$'s have different colors. Let $d_i$ be the size of
$N_i$, for each $i\in\{1,2,\cdots, l\}$. Suppose that $d_1\geq
d_2\geq\cdots\geq d_{l(v)}$ (ties break arbitrarily). In this case,
we say that $d_i$ is the {\it $i$-th degree of $v$}, and the color
of nodes in $N_i$ is the {\it $i$-th  color of neighbors of $v$},
for all $i\in\{1,2,\cdots,l\}$.

The length of degrees, the $i$-th degree and the $i$-th color of
neighbors of vertices have some interesting properties, including
the ones validated by experiments in ~\cite{LZPL2013a}: (i) The
length of degrees of a vertex is always bounded by $O(\log n)$, (ii)
The first degrees $d_1$'s are large, (iii) The second degrees are
always as small as constants, and (iv) For a vertex $v$, if the
length of degrees of $v$ is $l(v)>1$, then for any $i>1$, the $i$-th
color of neighbors of $v$ is distributed with a high degree of
randomness and uniformity. These properties are essential to the
experimental analysis of security of the networks
in~\cite{LZPL2013a}.

To theoretically prove the results, we define some useful notations.

\bigskip

\begin{definition} \label{def:degrees}
Let $G=(V,E)$ be a network constructed from the security model.
Given a node $v\in V$:

\begin{enumerate}
\item For every $j$, we define the $j$-th degree of $v$
at the end of time step $t$ to be the number of the $j$-th largest
set of homochromatic neighbors at the end of time step $t$, written
by $d_j(v)[t]$.

\item We define the $j$-th degree of $v$ to be the $j$-th degree of $v$ at
the end of  the construction of network $G$, written by $d_j(v)$.

\item We define the length of degrees of $v$ at the end of time step
$t$ to be the number of colors associated with neighbors of $v$ at
the end of time step $t$, written by $l(v)[t]$.

\item We define the length of degrees of $v$ to be the length of
degrees of $v$ at the end of the construction of $G$, written by
$l(v)$.
\end{enumerate}

\end{definition}

\smallskip

In sharp contrast to classic graph theory, for a network constructed
from our security model, $G$ say, and a vertex $v$ of $G$, $v$ has a
{\it priority of degrees}. This new feature must be universal in
real networks in the following sense: A community is an
interpretable object in a network such that nodes of the same
community share common features. In this case, a vertex $v$ may have
its own community and may link to some neighbor communities by some
priority ordering. In our model, a node $v$ more likes to contact
with nodes sharing the same color (or feature) with it, and has no
much preferences in contacting with nodes in its neighbor
communities.

\bigskip

\begin{definition} \label{def:degreepriority} (Degree Priority) Let
$v$ be a node of $G$ constructed from the security model created at
time step $t_0$, and $t\geq t_0$.

\begin{enumerate}

\item Suppose that $N_1, N_2, \cdots, N_l$ are all the homochromatic
neighbors of $v$ at the end of time step $t$ listed decreasingly by
the sizes of the sets $N_j$. For $d_j=|N_j|$ for each $j$, we say
that $(d_1, d_2,\cdots, d_l)$ is the degree priority of $v$ at the
end of time step $t$, written by $ dp(v)[t]=(d_1, d_2,\cdots, d_l)$.

\item We define the degree priority of $v$ in $G$ to be the degree
priority of $v$ at the end of the construction of $G$, written by $
dp(v)$.

\end{enumerate}

\end{definition}

\smallskip

The degree priority of nodes in $G$ satisfies some nice
probabilistic and combinatorial properties.

\bigskip

\begin{theorem} \label{thm:length} (Degree Priority Theorem) Let $G$ be a network
constructed from the security model with $d\geq 2$, and $a>1$. Then
with probability $1- o(1)$, for a randomly chosen node $v$, the
following properties hold:

\begin{enumerate}
\item The length of degrees of $v$ is bounded by $O(\log n)$, which
is an upper bound independent of $a$.
\item The first degree of $v$ is the number of neighbors that
share the same color as $v$.
\item The second degree of $v$ is bounded by $O(1)$, so that for any possible $j>1$, the $j$-th degree of $v$ is $O(1)$.
\item The first degree of a seed node is lower bounded by $\Omega (\log^{\frac{a+1}{4}} n)$.

\end{enumerate}

\end{theorem}

\smallskip

By (2), (3) and (4) of Theorem~\ref{thm:length}, we understand that
for a community $G_X$ induced by a homochromatic set $X$, the seed
node, $x_0$ say, of $X$ has a large first degree and constant second
degree, so that it is unlikely to be infected by a single neighbor
community, $G_Y$ say. Combining with (1), this ensures that for
properly chosen $a$, the seed node $x_0$ of $G_X$ is hard to be
infected by the collection of all its neighbor communities alone.
Such a community is regarded as a {\it strong community}.
Theorem~\ref{thm:Securityproperties} ensures that for properly
chosen $a$, almost all communities are strong, so that each of them
is hard to be infected by the collection of all its neighbor
communities alone.

Combining Theorem~\ref{thm:Securityproperties} and
Theorem~\ref{thm:length} gives us a better understanding for the
reasons why networks of the security model are secure. However, to
prove the security theorems, we have to understand the cascading
behaviors of attacks in the networks.

We define a community of $G$ is the induced subgraph of a
homochromatic set. We say that a community, $G_X$ say, is created at
time step $t$ if the seed node $x_0$ of $X$ is created at time step
$t$.

To understand the cascading behaviors, we define:

\bigskip

\begin{definition} \label{def:inj} Let $x$ and $y$ be two nodes of
$G$. We say that $x$ injures $y$, if the infection of $x$
contributes to the probability that $y$ becomes infected. Otherwise,
we say that $x$ fails to injure $y$.

\end{definition}

\smallskip

We will show that the infection of a community from a neighbor
community satisfies a number of combinatorial properties.

\bigskip

\begin{theorem} \label{thm:injury} (Infection-Inclusion Theorem) Suppose that $X$ and $Y$ are
 two homochromatic sets, and that $G_X$ and $G_Y$
are two communities. Let $x_0$, and $y_0$ be the seed nodes of $X$
and $Y$ respectively. Suppose that $x_0$ and $y_0$ are created at
time step $s$ and $t$ respectively. Then the injury of $G_Y$ from
community $G_X$ satisfies the following properties:

\begin{enumerate}

\item [(1)] If $s<t$, then

\begin{enumerate}

\item [(i)] The community $G_X$ created at time step $s$ fails to
injure any non-seed node in the community $G_Y$ created at time step
$t$.

\item [(ii)] The injury of the seed node $y_0$ created at time
step $t$ from the whole community created at time step $s$ is
bounded by a constant $O(1)$.

\end{enumerate}

\item [(2)] If $s>t$, then

\begin{enumerate}

\item [(i)] All the non-seed nodes in $G_X$ created at time step $s$
fail to injure any node in the community $G_Y$ created at time step
$t$.

\item [(ii)] The injury of the seed node created at time step $t$ from the community
created at time step $s$ is bounded by $1$.

\item [(iii)] The injury of a non-seed node in the community created at time step $t$
from the seed node created at time step $s$ follows the edge created
by step (3) (b) of Definition \ref{def:Securitymodel}.

\end{enumerate}

\item [(3)] The seed node $y_0$ of $G_Y$ created at time step $t$
can be injured only by:

\begin{enumerate}
\item [(i)] Communities created at time step $<t$.
\item [(ii)] The seed nodes of communities created at time step
$>t$.

\end{enumerate}

\item [(4)] A non-seed node $y$ of $G_Y$ created at time step $t$
can be injured only by seed nodes created at time step $>t$ through
the edge created by (3) (b) of Definition~\ref{def:Securitymodel}.

\end{enumerate}

\end{theorem}

\smallskip

(1), (2) and (3) of Theorem~\ref{thm:injury}, together with
Theorem~\ref{thm:length}, show furthermore that, a seed node, $v$
say, of $G$ are strong against infections from the collection of all
the communities other than its own community.

Suppose that $X$, $Y$ and $X$ are three homochromatic sets created
at time steps $t_1$, $t_2$ and $t_3$ respectively. Let $x_0$, $y_0$
and $z_0$ be the seed nodes of $X$, $Y$ and $Z$ respectively. It is
possible that $x_0$ infects a non-seed node $y_1$ of $Y$, $y_1$
infects all nodes in $Y$, including $y_0$, and $y_0$ infects a
non-seed node $z_1$ of $Z$. (4) of Theorem~\ref{thm:injury} ensures
that $t_1>t_2>t_3$, and that the edges $(x_0,y_1)$ and $(y_0,z_1)$
must be created by (3) (b) of Definition \ref{def:Securitymodel}.
The key point is that the edges $(x_0,y_1)$ and $(y_0,z_1)$ must be
embedded in a tree of height $O(\log n)$ which we will call the {\it
infection priority tree} (IPT, for short) $T$ of $G$. The infection
priority tree $T$ of $G$ is essentially a graph constructed by the
preferential attachment model with average number of edges $d'=1$,
which almost surely has height $O(\log n)$.

Therefore a targeted or infected strong community triggers at most
$O(\log n)$ many strong communities to be infected, by Theorem
\ref{thm:Securityproperties}, each community has size at most
$O(\log^{a+1}n)$. For any initial set of attacks $S$ of size
polynomial in $\log n$, suppose that every community which is not
strong has already been infected by attacks on $S$ automatically.
Let $K$ be the number of communities that are not strong. Then there
are at most $|S|+K$ strong communities trigger infections in the
infection priority tree $T$. This shows that there are at most
$O((|S|+K)\cdot\log n)$ communities in each of which there is at
least one node is infected by attacks on $S$. In this case, again by
Theorem \ref{thm:Securityproperties}, even if all the nodes in an
infected community are infected, the total number of infected nodes
is a negligible number comparing with the size of the network. This
sketch depends on an estimation of $K$, the number of communities
that are not strong, which will be given in the full proofs in later
sections.

Therefore (1), (2) and (4) of Theorem~\ref{thm:injury} ensure that
the infection of a non-seed node, $v$ say, is always one-way from a
seed node created late than $v$, following an edge in the infection
priority tree. By modulo the injury among the seed nodes, we are
able to show that the infections of non-seed nodes can only proceed
in the infection priority tree of height $O(\log n)$.

Now we fully understand that the combination of
Theorems~\ref{thm:Securityproperties}, \ref{thm:length}, and
\ref{thm:injury} does allow us to prove some security theorems of
the security model. This also explores the following security
principle of networks.

\bigskip

{\bf Security Principle}:

\begin{enumerate}
\item Small community phenomenon (by Theorem
\ref{thm:Securityproperties})
\item The number of seed nodes or hubs is large (by Theorem
\ref{thm:Securityproperties})
\item Almost all seed nodes (or hubs) are strong against infections from the
collection of all their neighbor communities alone (by Theorem
\ref{thm:length})
\item There exists an infection priority tree $T$ of $G$ such that
infection of non-seed nodes of a community from a neighbor community
can only be triggered by seed nodes of the neighbor community
through edges in the infection priority tree $T$ of $G$ (by Theorem
\ref{thm:injury})

\item The infection priority tree $T$ of $G$ has height $O(\log n)$
(to be proved in Subsection \ref{subsec:basiclemma})

\end{enumerate}

\smallskip

\section{Security Theorems}\label{sec:theorems}

In this section, we state the theorems and discuss the relationships
among the theorems.

By applying Theorems \ref{thm:Securityproperties}, \ref{thm:length}
and \ref{thm:injury}, we are able to prove that networks constructed
from the security model are secure against any attacks of small
sizes under both uniform and random threshold cascading failure
models.

For the uniform threshold cascading failure model, we have:

\begin{theorem} \label{thm:cascadeonSecurity} (Uniform threshold security theorem)
 Let $G$ be a graph constructed from
$\mathcal{S}(n,a,d)$ with $p_i=\log^{-a} i$ for homophyly exponent
$a>4$ and for $d\geq 4$. Let the threshold parameter
$\phi=O\left(\frac{1}{\log^b n}\right)$ for
$b=\frac{a}{2}-2-\epsilon$ for arbitrarily small $\epsilon>0$.

Then with probability $1-o(1)$ (over the construction of $G$), there
is no initial set of poly-logarithmic size which causes a cascading
failure set of non-negligible size. Precisely, we have that for any
constant $c>0$,
$$\Pr_{{}G\in_{\rm R}\mathcal{S}(n,a,d),\ G=(V,E)}\left[\forall
S\subseteq V,\ |S|=\lceil\log^c n\rceil,\ |{\rm
inf}_G^\phi(S)|=o(n)\right]=1-o(1),$$

\noindent where ${\rm inf}_G^{\phi}(S)$ is the infection set of $S$
in $G$ with uniform threshold $\phi$.
\end{theorem}

By Theorem~\ref{thm:cascadeonSecurity}, if $a>4$, and $d\geq 4$,
then for $\phi=O(1/\sqrt{\log^b n})$, networks constructed by the
security model $\mathcal{S}(n,a,d)$ are $\phi$-secure. Here $\phi$
is arbitrarily close to $0$, i.e., $\phi=o(1)$. Therefore, by
Definition~\ref{def:u-security}, for $a>4$, and $d\geq 4$, networks
in $\mathcal{S}(n,a,d)$ are secure under the uniform threshold
cascading failure model of attacks.

For the random threshold cascading failure model, each node $v$
picks randomly, uniformly and independently a threshold $\rho_v$
from $1,2,\ldots,d_v$.  Let ${\rm inf}_G^{\rm R}(S)$ be the
infection set of attacks on $S$ in $G$. We show that graphs
generated by $\mathcal{S}(n,a,d)$ are secure.

\begin{theorem} \label{thm:rancascadeonSecurity}
(Random threshold security theorem ) Let $a>6$ be the homophyly
exponent, and $d\geq 4$. Suppose that $G$ is a graph generated from
$\mathcal{S}(n,a,d)$.

Then with probability $1-o(1)$ (over the construction of $G$), there
is no initial set of poly-logarithmic size which causes a cascading
failure set of non-negligible size. Formally, we have that for any
constant $c>0$,
$$\Pr_{{}G\in_{\rm R}\mathcal{S}(n,a, d), \ G=(V,E)}\left[\forall
S\subseteq V,|S|=\lceil\log^c n\rceil,|{\rm inf}_G^{\rm
R}(S)|=o(n)\right]=1-o(1).$$
\end{theorem}

Theorems~\ref{thm:cascadeonSecurity} and
~\ref{thm:rancascadeonSecurity} show that for appropriately chosen
parameters, networks constructed from the security model are
provably secure for any attacks of small sizes under both uniform
and random threshold cascading failure models. By
Definitions~\ref{def:u-security}, \ref{def:r-security},
~\ref{def:security}, and by Theorems~\ref{thm:cascadeonSecurity} and
~\ref{thm:rancascadeonSecurity}, the security model in
Definition~\ref{def:Securitymodel} is secure.

The preferential attachment model was proposed to capture real
networks. It has become a classic model of networks. We use
$\mathcal{P}(n,d)$ to denote the set of random graphs of $n$ nodes
constructed from the PA model with average number of edges $d$.
Numerous experiments have shown that networks of the preferential
attachment model are insecure, see for instance
Figure~\ref{fig:_cascading_3models_N=10000_d=10_a=1.5}. Therefore
the best possible result we could look for would be the robustness
results for the PA model. People may take for granted that networks
of the PA model are robust, although there was no definition for
robustness in the literature. Here we have rigorous definition of
robustness of a model of networks, given in
Definitions~\ref{def:u-robustness}, \ref{def:r-robustness}, and
\ref{def:robustness}. This poses a fundamental question: Are
networks of the PA model really robust?

 We show
that, for large enough edge parameter $d$, for uniform threshold
cascading failure model, if the threshold is slightly less than
$1/d$, then just one randomly picked initial node is sufficient to
infect a significant fraction of the whole network with high
probability.

\begin{theorem} \label{thm:cascadeonPA} (Global cascading of a single
node in PA) For any $\varepsilon>0$, there exists a positive integer
$d_\varepsilon$ such that for any integer $d\geq d_\varepsilon$, if
$G=(V,E)$ is constructed from $\mathcal{P}(n,d)$, then with
probability $1-o(1)$ (over the construction of $G$), the following
inequality holds:
\[\Pr\limits_{v\in_{\rm R} V}[{\rm inf}_{G}^\phi(\{v\})=V]\geq\frac{2}{3}\left(1-\frac{1}{(1+\varepsilon)^2}\right),\]
where $\phi=\frac{1}{(1+\varepsilon)d}$.
\end{theorem}

Therefore if $\log n$ initial nodes are randomly picked, then the
whole graph $G$ will be infected with probability $1-o(1)$.

\begin{theorem} \label{cor:PApositive}(Global cascading theorem of PA)
For any $\varepsilon>0$, there exists a positive integer
$d_\varepsilon$ such that for any integer $d\geq d_\varepsilon$, for
threshold parameter $\phi=\frac{1}{(1+\varepsilon)d}$,

$$\Pr_{{}S\subset_{\rm R} V,\ |S|=\log n}[{\rm
inf}_{G}^\phi(\{S\})=V]=1-o(1).$$
\end{theorem}
\begin{proof} By Theorem \ref{thm:cascadeonPA}.
\end{proof}

Consequently, $\mathcal{P}(n,d)$ is not $\phi$-robust for all
$\phi\leq \frac{1}{(1+\varepsilon)d}$. By
Definitions~\ref{def:u-robustness} and \ref{def:robustness}, and by
Theorem~\ref{cor:PApositive}, the preferential attachment model is
not robust. In fact, each of the nontrivial networks constructed
from the PA model is non-robust. This result shows that if real
networks truthfully follow the PA model, then the networks would be
not only insecure, but also unavoidably non-robust. This makes the
situation even worse in practical applications, because, a few or
even one random error may cause a global cascading failure of the
whole network.

On the other hand, we also show that if the threshold is larger than
$1/d$, then with probability $1-o(1)$, $o(\sqrt{n})$ randomly picked
initially infected nodes are insufficient to infect even one more
node, and the PA model is robust in this case. In fact, we are able
to prove a stronger result that holds for arbitrarily given simple
(or almost simple) graphs \footnote{A simple graph is a graph having
no multi-edge and self-loop.}.

\begin{theorem} \label{thm:negative} (Robustness theorem of graphs)
Given a simple graph $G=(V,E)$ whose nodes have minimum degree $d$.
Let $n=|V|$ and $d$ be a constant independent of $n$. Let
$\phi>\frac{l}{d}$, where $l$ is an integer from the interval
$[1,d-1]$. Let $S\subseteq V$ be a randomly picked subset of size
$k=o(n^{\frac{l}{l+1}})$. Then
\[\Pr\limits_{S\subset_{\rm R} V}[{\rm inf}_{G}^\phi(S)=S]=1-o(1).\]
\end{theorem}

By using this, we have:

\begin{theorem} \label{cor:PAnegative} (Robustness theorem of PA)
For any integer $d\geq 2$ and $\phi>\frac{1}{d}$, $\mathcal{P}(n,d)$
is $\phi$-robust.
\end{theorem}

\begin{proof}
Since the number of multi-edges and self-loops in $\mathcal{P}(n,d)$
is at most $O(\log n/n)$ (with probability almost $1$), the
probability that, in randomly picked $n^\lambda$ ($\lambda\leq 1/2$)
nodes, there is a node associating to some multi-edge or self-loop
is upper bounded by $o(1)$. It is easily observed that the result is
a straightforward corollary of Theorem \ref{thm:negative} in the
case of $l=1$.
\end{proof}

Theorem~\ref{cor:PAnegative} implies that for a network constructed
from the PA model, if every node has a threshold $\geq\phi$ for some
large constant $\phi$, then the network is robust against random
errors ( of small sizes).

By Theorems~ \ref{cor:PApositive} and ~\ref{cor:PAnegative}, the
value $1/d$ is a key {\it threshold} for the robustness of the PA
model. The two theorems characterize the robustness of networks of
the PA model under uniform threshold cascading failure model,
leaving open for the case of $\phi=1/d$. This clarifies the
experimental results of robustness of networks of the PA model.

The remaining sections are devoted to proofs of
Theorems~\ref{thm:Securityproperties}, \ref{thm:length},
\ref{thm:injury}, ~\ref{thm:cascadeonSecurity},
~\ref{thm:rancascadeonSecurity}, ~\ref{thm:cascadeonPA}
and~\ref{thm:negative}. In section~ \ref{sec:pro}, we prove Theorem
~\ref{thm:Securityproperties}. In Section~\ref{sec:pcomb}, we prove
Theorems~\ref{thm:length} and \ref{thm:injury}. In Section
~\ref{sec:cascadeonSecurity}, we prove
Theorems~\ref{thm:cascadeonSecurity}, and
~\ref{thm:rancascadeonSecurity} by using
Theorems~\ref{thm:Securityproperties}, ~\ref{thm:length} and
\ref{thm:injury}. In Section~ \ref{sec:cascade in PA}, we prove the
threshold theorem of robustness of networks of the PA model,
consisting of Theorems ~\ref{thm:cascadeonPA} and
~\ref{thm:negative}. In Section~\ref{sec:over}, we extend the
security model to high dimensions so that a node has $k$ colors for
$k>1$. In this case, communities in the network are overlapping. We
show that overlapping communities undermine security of networks. In
Section \ref{sec:conc}, we summarize the conclusions and discuss
some future directions.

\section{The Fundamental Theorem of the Security Model}\label{sec:pro}

In this section, we prove Theorem ~\ref{thm:Securityproperties}.
Before proving the theorem, we state the Chernoff bound below which
will be frequently used in our proofs.

\begin{lemma} (Chernoff bound, \cite{C81}) Let $X_1,\ldots,X_n$ be independent random variables
with $\Pr[X_i=1]=p_i$ and $\Pr[X_i=0]=1-p_i$. Denote the sum by
$X=\sum_{i=1}^n X_i$ with expectation $E(X)=\sum_{i=1}^n p_i$. Then
we have
$$\Pr[X\leq E(X)-\lambda]\leq \exp\left(-\frac{\lambda^2}{2E(X)}\right),$$
$$\Pr[X\geq E(X)+\lambda]\leq \exp\left(-\frac{\lambda^2}{2(E(X)+\lambda/3)}\right).$$
\end{lemma}

Let $G$ be a network constructed from the security model. We now
prove Theorem ~\ref{thm:Securityproperties}. We will prove (1), (2),
(3) and (4) of Theorem~\ref{thm:Securityproperties} in Subsections
\ref{sub3.0(0)}, \ref{sub3.1(1)}, \ref{sub3.1(2)} and
\ref{sub3.1(3)} respectively.

\subsection{Basic Properties}
\label{sub3.0(0)}

 In this subsection, we prove (1) of Theorem~ \ref{thm:Securityproperties}.
It consists of two results, the first is the estimation of number of
seed nodes, and the second is the upper bound of sizes of the
homochromatic sets.

\begin{proof} (Proof of (1) of Theorem~
\ref{thm:Securityproperties}) We use $G[t]$ to denote the graph
constructed at the end of time step $t$ of the construction of $G$.
Let $T_1=\log^{a+1} n$, and $C_t$ be the set of all colors appear in
$G[t]$.

\bigskip

For (i). It suffices to show that the size of $C_t$ is bounded as
desired. For this, we have:

\bigskip

\begin{lemma} \label{lem:colorsize}
With probability $1-o(1)$, for all $t\geq T_1$, $\frac{t}{2\log^a
t}\leq |C_t|\leq \frac{2t}{\log^a t}$.
\end{lemma}

\begin{proof}
The expectation of $|C_t|$ is
$$E[|C_t|] = 2+ \sum_{i=3}^{t}\frac{1}{\log^a i}.$$
By indefinite integral
$$\int (\frac{1}{\log^a x}-\frac{a}{\log^{a+1}x}) dx=\frac{x}{\log^a x}+C,$$
we know that if $t$ is large enough, then
\begin{eqnarray*}
\sum_{i=3}^t \frac{1}{\log^a i} & \leq &
1+\int_2^t \frac{1}{\log^a x}dx\\
& \leq & \int_2^t \frac{6}{5}(\frac{1}{\log^a
x}-\frac{a}{\log^{a+1}x}) dx\\
& \leq & \frac{4t}{3\log^a t},
\end{eqnarray*}
where $\frac{6}{5}$ and $\frac{4}{3}$ are chosen arbitrarily among
the numbers larger than $1$. Similarly,
\begin{eqnarray*}
\sum_{i=3}^t \frac{1}{\log^a i} & \geq &
\int_2^t \frac{1}{\log^a x}dx\\
& \geq & \int_2^t \frac{5}{6}(\frac{1}{\log^a
x}-\frac{a}{\log^{a+1}x}) dx\\
& \geq & \frac{3t}{4\log^a t}.
\end{eqnarray*}

By the Chernoff bound and the fact that $t\geq T_1=(\log n)^{a+1}$,
with probability $1-exp(-\Omega(\frac{t}{\log^a t}))=1-o(n^{-1})$,
we have $\frac{t}{2\log^a t}\leq |C_t|\leq \frac{2t}{\log^a t}$. By
the union bound, such an inequality holds for all $t\geq T_1$ with
probability $1-o(1)$.
\end{proof}
(i) follows from Lemma~\ref{lem:colorsize}.

\smallskip

Lemma~\ref{lem:colorsize} depends on only the probability
$p_i=1/(\log i)^a$ with which the node created at time step $i$
chooses a new color. It is a useful fact throughout the proofs, from
which we define:

 \bigskip

\begin{definition} \label{def:event} We define
 $\mathcal{E}$ to be the event that $|C_t|$ is bounded in the
interval $\left[\frac{t}{2\log^a t}, \frac{2t}{\log^a t}\right]$.

\end{definition}
\smallskip

 By Lemma
\ref{lem:colorsize}, almost surely, the event $\mathcal{E}$  holds
for all $t\geq T_1$.

\bigskip

For (ii). We estimate the size of all the homochromatic sets.

\bigskip

\begin{lemma} \label{lem:homosetsizeupper}
With probability $1-o(1)$, the following properties hold:

(1) Every community has size bounded by $O(\log^{a+1}n)$, and

(2) For every $t\geq T_1$, every community at the end of time step
$t$ has size bounded by $O(\log^{a+1}t)$.

\end{lemma}

\begin{proof} For (1). It suffices to show that with probability $1-o(n^{-1})$, the
homochromatic set of the first color $\kappa$ has size
$O(\log^{a+1}n)$.

We define an indicator random variable $Y_t$ for the event that the
vertex created at time $t$ chooses color $\kappa$. We also define
$\{Z_t\}$ to be the independent Bernoulli trails such that
$$\Pr[Z_t=1]=\left(1-\frac{1}{\log^a n}\right)\frac{2\log^a
t}{t}.$$ Conditioned on the event $\mathcal{E}$, we know that
$Y:=\sum_{t=1}^n Y_t$ is stochastically dominated by
$Z:=\sum_{t=1}^n Z_t$. The latter has an expectation
$$E[Z]\leq \sum_{t=1}^n\frac{2\log^a
t}{t}\leq 2\log^{a+1} n.$$ By the Chernoff bound,
$$\Pr[Z>4\log^{a+1} n] \leq n^{-1}.$$
Therefore, with probability $1-n^{-1}$, the size of $S_\kappa$ is
$Y\leq 4\log^{a+1}n$. (1) follows.

For (2). This follows from the proof of (1) above. (2) holds.

Lemma~\ref{lem:homosetsizeupper} follows.
\end{proof}

(ii) holds.

This proves (1) of Theorem \ref{thm:Securityproperties}.

\subsection{Power Law}
\label{sub3.1(1)}

In this subsection, we probe (2) of Theorem
\ref{thm:Securityproperties}, consisting of power law of the induced
subgraph of communities, of the degree distributions of the
homochromatic sets, and of the whole network $G$.

Before proving the results, we first prove both a lower bound and an
upper bound for the sizes of well-evolved communities.

Recall that $T_1=\log^{a+1} n$. Let $T_2=(1-\delta_1)n$, for
$\delta_1=\frac{10}{\log^{a-1}n}$. We have:

\bigskip

\begin{lemma}\label{lem:setsize}
With probability $1-o(1)$, both (1) and (2) below hold in $G$:

\begin{enumerate}
\item [(1)] For a community created at a time step $\leq T_2$, it has size at least $\log
n$;
\item [(2)] For a community created at a time step $>T_2$, it has size at most $30\log n$.
\end{enumerate}
\end{lemma}

\begin{proof}

For (1). We only need to prove that, on the condition of event
$\mathcal{E}$ in Definition~\ref{def:event}, any homochromatic set
$S_\kappa$ created before time step $T_2+1$ has size at least $\log
n$ with probability $1-o(n^{-1})$.

For every $t>T_2$, let $Y_t$ be the indicator random variable that
the vertex, $v$ say, created at time step $t$ chooses old color
$\kappa$. For $t>T_2$, let $\{Z_t\}$ be the independent Bernoulli
trails such that
\begin{eqnarray}
\Pr[Z_t=1]=(1-\frac{1}{\log^a (1-\delta_1)n})\frac{\log^a
t}{2t}.\label{eqn:randomv1}
\end{eqnarray}

Conditioned on the event $\mathcal{E}$, we know that $Y:=\sum_{t\geq
T_2+1}^n Y_t$ stochastically dominates $Z:=\sum_{t\geq T_1+1}^n
Z_t$, which has expectation
$$E[Z]\geq \sum_{t=T_2+1}^n\frac{\log^a t}{2t}\geq \frac{\delta_1}{2}\log^a (1-\delta_1)n\geq 4\log n.$$
By the Chernoff bound,
\begin{eqnarray*}
\Pr[Z<\log n] \leq e^{-\frac{3^2\log n}{2\times
4}}=n^{-\frac{9}{8}}.
\end{eqnarray*}
Thus, with probability $1-o(n^{-1})$, the size of $S_\kappa$ is at
least $\log n$.

For (2). The proof is similar to that of (1) above. We only need to
prove that, on the condition of event $\mathcal{E}$, any
homochromatic set $S_\kappa$ created after $T_2$ has size at most
$30\log n$ with probability $1-o(n^{-1})$. For $t>T_2$, we consider
the Bernoulli random variables $\{Z_t\}$ defined by
\begin{eqnarray}
\Pr[Z_t=1]=(1-\frac{1}{\log^a n})\frac{2\log^a
t}{t}.\label{eqn:randomv2}
\end{eqnarray}
Note that
$$E[Z]\leq \sum_{t=T_2+1}^n\frac{2\log^a t}{t}\leq \frac{2\delta_1}{1-\delta_1}\log^a n.$$
By a similar analysis to that in (1) above, we know that with
probability $1-o(n^{-1})$, the size of $S_\kappa$ is at most $30\log
n$.
\end{proof}

\smallskip

The proof of Lemma~\ref{lem:setsize} depends on both the probability
$1-p_i$ with which the newly created node chooses an old color, and
the randomness and uniformity of the choice of the old color at time
step $i$ for all $i$'s.

By Lemma \ref{lem:setsize}, we know that each of the communities
born before time step $T_2+1$ has expected size $\omega(1)$, and
that all the communities born at time steps $\leq T_2$ account for
($1-o(1)$) of all the communities. Therefore we prove the power law
distribution only for the communities born at time steps $\leq T_2$.

\bigskip

 For both (i) and (ii). Now we turn to prove two results:

\begin{enumerate}
 \item [(A)] For each homochromatic set $X$, the degrees of nodes in $X$
 follow a power law, and

 \item [(B)] For each homochromatic set $X$, the induced subgraph $G_X$ of $X$ follow a power law.
\end{enumerate}

We prove both (A) and (B) together. We consider only the non-trivial
homochromatic sets, i.e., the well-evolved communities, by ignoring
the few most recently created communities.

By (4) of Definition \ref{def:Securitymodel}, each community
basically follows the classical preferential attachment model, we
are able to give explicit expressions for the expected numbers of
nodes of degree $k$ for all $k$, for each of the homochromatic sets
and for the induced subgraphs of the homochromatic sets.

In fact, as we will show below that the contribution to the degrees
of a homochromatic set from the global edges is much more smaller
than that from the local edges of the homochromatic set. This is the
key point to our proofs of the power law of almost all the
communities.

We use $X$ to denote a homochromatic set of a fixed color, $\kappa$
say. Let $T_0$ be the time step at which $X$ is created.

For positive integers $s$ and $k$, we define $A_{s,k}$ to be the
number of nodes of degree $k$ in $X$ when $|X|$ reaches $s$,
$B_{s,k}$ to be the number of nodes of degree $k$ in the induced
subgraph of $X$ when $|X|$ reaches $s$, and $g_{s,k}$ to be the
number of global edges associated with the nodes in $X$ of degree
$k$ in the induced subgraph of $X$ when $|X|$ reaches $s$. By
definition, we have $A_{1,d}=1$ and $A_{1,k}=0$ for all $k>d$, and
$B_{1,k}=0$ for all $k$. We also have $A_{s,k}=B_{s,k}+g_{s,k}$.
Then we establish the recurrence formula for the expectations of
both $A_{s,k}$ and $B_{s,k}$.

Firstly, we define some notations associated with $X$ and its size
$|X|$:

-- we use $T(s)$ (or $T$, for simplicity)  to denote the time step
at which the size of $X$ becomes to be $s$,

-- we use $s_1$ to denote the number of global edges connecting to
$X$ in the case that $|X|=s$.

We consider the time interval $(T(s-1), T(s))$. Then the number of
times that a global edge is created  and linked to a node in $X$ of
degree $k$ at some time step in the interval $(T(s-1), T(s))$ is
expected to be $\Theta(\frac{1}{\log^a T}\cdot\frac{k\cdot
A_{s,k}}{2dT}/\frac{\log^a T}{T})=\Theta(\frac{k\cdot
A_{s,k}}{\log^{2a}T})$. Denote $\Theta(\log^{2a}T)$ by $s_2$.

Then for $s>1$ and $k>d$, we have
$$E(A_{s,k})=A_{s-1,k}\left(1-\frac{kd}{2d(s-1)+s_1}-\frac{k}{s_2}\right)
+A_{s-1,k-1}\cdot\left(\frac{(k-1)d}{2d(s-1)+s_1}+\frac{k-1}{s_2}\right)+O\left(\frac{1}{s^2}\right).$$
Taking expectations on both sides, we have
\begin{eqnarray} \label{eqn:A_krecurrence}
E(A_{s,k}) &=&
E(A_{s-1,k})\left(1-\left(\frac{1}{2(s-1)+s_1/d}-\frac{1}{s_2}\right)k\right)\nonumber \\
&&+E(A_{s-1,k-1})\left(\frac{1}{2(s-1)+s_1/d}+\frac{1}{s_2}\right)(k-1)+O\left(\frac{1}{s^2}\right).
\end{eqnarray}
If $k=d$, then
\begin{equation} \label{eqn:A_drecurrence}
E(A_{s,d})=E(A_{s-1,d})\left(1-\left(\frac{1}{2(s-1)+s_1/d}-\frac{1}{s_2}\right)d\right)+1+O\left(\frac{1}{s^2}\right).
\end{equation}

Similarly, for $s>1$ and $k>d$,
$$E(B_{s,k})=B_{s-1,k}-\frac{d\cdot(kB_{s-1,k}+g_{s-1,k})}{2d(s-1)+s_1}+\frac{d\cdot((k-1)B_{s-1,k-1}+g_{s-1,k-1})}{2d(s-1)+t_1}+O(\frac{1}{s^2}).$$
Taking expectations on both sides, we have
\begin{eqnarray} \label{eqn:B_krecurrence}
E(B_{s,k}) &=& E(B_{s-1,k})\left( 1-\frac{kd}{2d(s-1)+s_1}
\right)+E(B_{s-1,k-1})\cdot
\frac{(k-1)d}{2d(s-1)+s_1}\nonumber \\
&&+\frac{E(g_{s-1,k-1}-g_{s-1,k})}{2d(s-1)+s_1}+O(\frac{1}{s^2}).
\end{eqnarray}
If $k=d$, then
\begin{eqnarray} \label{eqn:B_drecurrence}
E(B_{s,d}) &=&
B_{s-1,d}-\frac{d\cdot(dB_{s-1,d}+g_{s-1,d})}{2d(s-1)+s_1}+1+O(\frac{1}{s^2})\nonumber \\
&=& B_{s-1,d}\left( 1-\frac{d}{2(s-1)+s_1/d} \right)+\left(
1-\frac{g_{s-1,d}}{2d(s-1)+s_1} \right),
\end{eqnarray}
and
$$E(B_{s,d})=E(B_{s-1,d})\left( 1-\frac{d}{2(s-1)+s_1/d} \right)+\left(
1-\frac{E(g_{s-1,d})}{2d(s-1)+s_1} \right).$$

To solve the recurrences, we invoke the following lemma.

\begin{lemma}
\label{lem:recurrence} (\cite{CL06}, Lemma 3.1) Suppose that a
sequence $\{a_s\}$ satisfies the recurrence relation
$$a_{s+1}=(1-\frac{b_s}{s+s_1})a_s+c_s~~{\it for}~~s\geq s_0,$$
where the sequences $\{b_s\},\{c_s\}$ satisfy
$\lim_{s\rightarrow\infty}b_s=b>0$ and
$\lim_{s\rightarrow\infty}c_s=c$ respectively. Then the limitation
of $\frac{a_s}{s}$ exists and
$$\lim_{s\rightarrow\infty}\frac{a_s}{s}=\frac{c}{1+b}.$$
\end{lemma}

For the recurrence of $E(A_{s,k})$, by Lemma~\ref{lem:setsize}, as
$n$ goes to infinity, $t=\omega(1)$ also goes to infinity. By the
definition of $s_2$, $s_2=\Theta(\log^{2a}T)=\omega(s)$.

To deal with $s_1$, we give a upper bound for the expected volume of
$X$ at time $T$, denoted by $V_T$, as follows.
\begin{eqnarray*}
E(V_T) &\leq& \sum\limits_{i=2}^T \left[ \left( 1-\frac{1}{\log^a i}
\right)\cdot\frac{2d}{|C_i|}+\frac{1}{\log^a i}\cdot\frac{d
V_{i-1}}{2di} \right]\\
&\leq& \sum\limits_{i=2}^T \frac{2d}{|C_i|} =
O\left(\sum\limits_{i=2}^T \frac{4d\log^a i}{i}\right) = O(\log^a
T).
\end{eqnarray*}

So it is easy to observe that $\frac{s_1}{t}=O\left( \frac{1}{\log^a
T}\cdot\frac{V_T}{2dT}/ \frac{\log^a T}{T} \right)=O\left(
\frac{1}{\log^{a-1} T} \right)$ goes to zero as $s$ approaches to
infinity.

For the recurrence of $E(B_{s,k})$, we show that as $s$ goes to
infinity, both $\frac{E(g_{s-1,k-1}-g_{s-1,k})}{2d(s-1)+s_1}$ and
$\frac{E(g_{s-1,d})}{2d(s-1)+s_1}$ approach to $0$. Define
$g_s=\sum_i g_{s,i}$ to be the total number of global edges
associated to $X$ when $|X|$ reaches $s$. We only have to show that
$E(\frac{g_s}{s})\rightarrow 0$ as $s\rightarrow \infty$.

Suppose that the seed node of $X$ is created at time $T_0$.
\begin{eqnarray*}
E(g_s) = O\left( \sum\limits_{i=T_0}^{T(s)}\frac{1}{\log^a
i}\cdot\frac{V_i}{2di} \right) = O\left(\sum\limits_{i=T_0}^{T(s)}
\frac{\log i}{2di}\right) = O(\log^2 T(s)-\log^2 T_0).
\end{eqnarray*}

Note that when we consider the size of $X$ at sometime $t>T_0$, we
have
\begin{eqnarray*}
E(|X|) &=& \sum\limits_{i=T_0}^{t} \left( 1-\frac{1}{\log^a
i}\cdot\frac{1}{|C_i|} \right) = \Omega\left(
\sum\limits_{i=T_0}^{t}\frac{\log^a i}{2i} \right)\\
&=& \Omega\left( \int_{T_0}^{t}\frac{\log^a x}{2x}dx \right) =
\Omega(\log^{a+1}t-\log^{a+1}T_0).
\end{eqnarray*}
Thus at time $T(s)$, by the Chernoff bound, with probability
$1-o(1)$, $s=\Omega(\log^{a+1}T(s)-\log^{a+1}T_0)$. Therefore,
$E(g_s)=o(s)$, that is, $E(\frac{g_s}{s})\rightarrow 0$ as
$s\rightarrow \infty$.

Then we turn to consider the recurrences of $E(A_{s,k})$ and
$E(B_{s,k})$. The terms $s_1/d$ and $\frac{1}{s_2}$ in equalities
(\ref{eqn:A_krecurrence}) and (\ref{eqn:A_drecurrence}) are
comparatively negligible. The terms
$\frac{E(g_{s-1,k-1}-g_{s-1,k})}{2d(s-1)+s_1}$ and
$\frac{E(g_{s-1,d})}{2d(s-1)+s_1}$ in equalities
(\ref{eqn:B_krecurrence}) and (\ref{eqn:B_drecurrence}),
respectively, are also comparatively negligible. By Lemma
\ref{lem:recurrence}, $\frac{E(A_{s,k})}{s}$ and
$\frac{E(B_{s,k})}{s}$ must have the same limit as $t$ goes to
infinity. Next, we will only give the proof of the power law
distribution for $E(A_{s,k})$, which also holds for $E(B_{s,k})$.

Denote by $S_k=\lim_{t\rightarrow\infty}\frac{E(A_{s,k})}{s}$ for
$k\geq d$. In the case of $k=d$, we apply Lemma \ref{lem:recurrence}
with $b_s=\frac{d}{2}$, $c_s=1+O(\frac{1}{s^2})$, $s_1=-1$, and get
$$S_d=\lim_{s\rightarrow\infty}\frac{E(A_{s,d})}{t}=\frac{1}{1+\frac{d}{2}}=\frac{2}{2+d}.$$
For $k>d$, assume that we already have
$S_{k-1}=\lim_{t\rightarrow\infty}\frac{E(A_{s,k-1})}{t}$. Applying
Lemma \ref{lem:recurrence} again with $b_s=\frac{k}{2}$,
$c_s=\frac{E(A_{s-1,k-1})}{s-1}\cdot\frac{k-1}{2}$, $s_1=-1$, we get
$$S_k=\lim_{t\rightarrow\infty}\frac{E(A_{s,k})}{s}=\frac{S_{k-1}\cdot\frac{k-1}{2}}{1+\frac{k}{2}}=S_{k-1}\cdot\frac{k-1}{k+2}.$$
Thus recurrently, we have
\begin{equation} \label{S_k expression}
S_k=S_d\cdot\frac{(d+2)!(k-1)!}{(d-1)!(k+2)!}=\frac{2d(d+1)}{k(k+1)(k+2)}.
\end{equation}
This implies
$$|E(A_{s,k})-S_k \cdot s|=o(s),$$
and thus
$$E(A_{s,k})=(1+o(1))k^{-3}s.$$
Since $s=\omega(1)$ goes to infinity as $n\rightarrow\infty$,
$E(A_{s,k})\propto k^{-3}$. For the same reason, $E(B_{s,k})\propto
k^{-3}$. This proves (A) and (B), and also completes the proof of
both (i) and (ii).

\bigskip

For (iii). For the whole network, a key observation is that the
union of several power law distributions is also a power law
distribution if the powers are equal. We will give the same explicit
expression of the expectation of the number of degree $k$ nodes by
combining those for the homochromatic sets, leading to a similar
power law distribution.

To prove the power law degree distribution of the whole graph, we
take the union of distributions of all homochromatic sets. We will
show that with overwhelming probability, almost all nodes belong to
some large homochromatic sets so that the role of small
homochromatic sets is negligible.

Suppose that $G$ has $m$ homochromatic sets of size at least $\log
n$. For $i=1,\ldots,m$, let $M_i$ be the size of the $i$-th
homochromatic set and $N_{s,k}^{(i)}$ denote the number of nodes of
degree $k$ when the $i$-th set has size $s$. For each $i$, we have
$$\lim_{n\rightarrow\infty}\frac{E(N_{M_i,k}^{(i)})}{M_i}=S_k.$$
Hence,

$$\lim_{n\rightarrow\infty}\frac{E(\sum_{i=1}^m
N_{M_i,k}^{(i)})}{\sum_{i=1}^m M_i}=S_k.$$

Let $M_0$ denote the size of the union of all other homochromatic
sets of size less than $\log n$, and $N_{s,k}^{(0)}$ denote the
number of nodes of degree $k$ in this union when it has size $s$. By
Lemma \ref{lem:setsize}, with probability $1-o(1)$, all these sets
are created after time $T_2$, and thus $M_0\leq n-T_2=
\frac{10n}{\log^{a-1}n}=o(n)$.

Define $N_{t,k}$ to be the number of nodes of degree $k$ in $G_t$,
that is, the graph obtained after time step $t$. Then we have
$$\lim_{n\rightarrow\infty}\frac{E(N_{n,k})}{n}=\lim_{n\rightarrow\infty}\frac{E(\sum_{i=0}^m N_{M_i,k}^{(i)})}{\sum_{i=0}^m M_i}.$$
For $M_0$, we have that
$$\lim_{n\rightarrow\infty} \frac{M_0}{\sum_{i=1}^m M_i}=\lim_{n\rightarrow\infty} \frac{M_0}{n-M_0}=0$$
and
$$\lim_{n\rightarrow\infty} \frac{E(N_{M_0,k}^{(0)})}{n} \leq \lim_{n\rightarrow\infty} \frac{M_0}{n}=0$$
hold with probability $1-o(1)$. So
$$\lim_{n\rightarrow\infty} \frac{E(N_{n,k})}{n}=\lim_{n\rightarrow\infty}\frac{E(\sum_{i=1}^m N_{M_i,k}^{(i)})}{\sum_{i=1}^m M_i}=S_k.$$
This implies
$$|E(N_{n,k})-S_k\cdot n|=o(n),$$
and thus,
$$E(N_{n,k})=(1+o(1))k^{-3}n,$$
and $E(N_{n,k})\propto k^{-3}$. (iii) follows.

This completes the proof Theorem~ \ref{thm:Securityproperties} (2).
\end{proof}

\subsection{Small World Property}
\label{sub3.1(2)}

For Theorem~ \ref{thm:Securityproperties} (3). Now we turn to prove
the properties of small diameters of each homochromatic set and
small world phenomenon of networks of the security model.

\smallskip

For (i). The diameter of the standard PA model is well-known
\cite{BR04}, where it has been shown that a randomly constructed
graph from the PA model, written $\mathcal{G}(n,d)$, has a diameter
$O(\log n)$ with probability $1-O(\frac{1}{\log^2 n})$.

(i) follows immediately from Theorem \ref{thm:Securityproperties}
(1) (ii).

\smallskip

For (ii). Now we prove the small world phenomenon. We adjust the
parameters in the proof of the PA model in~\cite{BR04} to get a
weaker bound on diameters, but a tighter probability. In so doing,
we have the following lemma.

\begin{lemma}\label{lem:padiam}
For any constant $a'>2$, there is a constant $K$ such that with
probability $1-\frac{1}{n^{a'+1}}$, a randomly constructed graph $G$
from the PA model $\mathcal{P}(n,d)$ has a diameter $Kn^{1/(a'+1)}$.
\end{lemma}

\begin{proof}
By a standard argument as that in the proof of the small diameter
property of networks of the preferential attachment.
\end{proof}

Moreover, to estimate the distances among seed nodes, we recall a
known conclusion on random recursive trees. A random recursive tree
is constructed by stages, at each stage, one new vertex is created.
A newly created node must be linked to an earlier node chosen
according to a uniform choice. In this case, we call it a uniform
recursive tree~\cite{MS95}. We use a result of Pittel in~\cite{P94},
saying that the height of a uniform recursive tree of size $n$ is
$O(\log n)$ with high probability.

\begin{lemma} \label{lem:resursivetrees}
(\cite{P94}) With probability $1-o(1)$, the height of a uniform
recursive tree of size $n$ is asymptotic to $e\log n$, where $e$ is
the natural logarithm.
\end{lemma}

To estimate the average node-to-node distance of $G$, we assume that
there are $m$ homochromatic sets of size at most $\log n$. Choose
$a'$ in Lemma~\ref{lem:padiam} to be the homophyly exponent $a$, and
then we have a corresponding $K$.

Given a homochromatic set $S$, we say that $S$ is {\it bad}, if the
diameter of $S$ is larger than $K|S|^{1/(a+1)}$.

We define an indicator $X_S$ of the event that $S$ is bad. Since
$\log n\leq|S|=O(\log^{a+1}n)$, by Lemma \ref{lem:padiam}, we have
$$\Pr[X_S=1]\leq \frac{1}{\log^{a+1}n}.$$

By Lemma \ref{lem:colorsize}, the expected number of bad sets is at
most $\frac{2n}{\log^a
n}\cdot\frac{1}{\log^{a+1}n}=\frac{2n}{\log^{2a+1}n}$. By the
Chernoff bound, with probability $1-O(n^{-2})$, the number of bad
sets is at most $\frac{3n}{\log^{2a+1}n}$. Thus the total number of
nodes belonging to some bad set is at most $\frac{3n}{\log^a n}$. On
the other hand, for any large set $S$ that is not bad, its diameter
is at most $K|S|^{1/(a+1)}=O(\log n)$.

Given two nodes $u$ and $v$ with distinct colors. Suppose that $c_0$
nd $c_1$ are the colors of $u$ and $v$ respectively, that $X$ and
$Y$ are the sets of nodes of colors $c_0$ and $c_1$ respectively,
and that $u_0$ and $v_0$ are the seed nodes in $X$ and $Y$
respectively. We consider a path from $u$ to $v$ as follows: (a) the
first part is a path from $u$ to $u_0$ within the induced subgraph
of $X$, (b) the second part is a path from $u_0$ to $v_0$ consisting
of only global edges, and (c) the third part is a path from $u_0$ to
$u$, consisting of edges in the induced subgraph of $Y$. By the
argument above, the number of the union of all bad homochromatic
sets is bounded by $O(\frac{n}{\log^a n})$.  By
Definition~\ref{def:Securitymodel}, the giant connected component of
all the seed nodes can be interpreted as a union of $d$ uniform
recursive trees. By lemma \ref{lem:resursivetrees}, with probability
$1-o(1)$, there is a path from $u_0$ to $v_0$ in the induced
subgraph of all seed nodes with length at most $O(\log n)$.
Combining the three paths in (a), (b) and (c) above, we know that
the average node to node distance in $G$ is at most
$O(\frac{\frac{2n^2}{\log^a n}\cdot\log^{a+1}n+n^2\cdot\log
n}{n^2})=O(\log n)$. (ii) follows.

\bigskip

For (iii). Suppose that $G$ is a network constructed from the
security model. We interpret $G$ as a directed graph as follows: For
an edge $(u,v)$ in $G$, if $u$ and $v$ are created at time steps $i,
j$ respectively, then for $i>j$, we identify the edge $(u,v)$ as a
directed edge $(i,j)$.

We give an algorithm as follows: For any two nodes $u$, $v$ in $G$,

\begin{enumerate}
\item Following the direction of time order in $G$ (that is, an edge $(x,y)$ means that $y$ is
 created earlier than $x$) to find the seed nodes of the homochromatic sets of $u$ and $v$,
$u_0$ and $v_0$ say, respectively.
\item Take random walks from  $u_0$ and $v_0$ in a directed uniform
recursive tree of all the seed nodes created in (3) (c) of
Definition~\ref{def:Securitymodel}, until the two random walks
cross.

\end{enumerate}

By (i), step (1) runs in time $O(\log\log n)$, by
Lemma~\ref{lem:resursivetrees}, step (2) runs in time $O(\log n)$.
(iii) follows.

This completes the proof of Theorem~\ref{thm:Securityproperties}
(3).

\subsection{Small Community Phenomenon}
\label{sub3.1(3)}

Before proving (4) of Theorem~\ref{thm:Securityproperties}, we
introduce some notations.

Let $X$ be a homochromatic set, and $x_0$ be the seed node of $X$.
We say that $X$ is created at time step $t$, if the seed node $x_0$
of $X$ is created at time step $t$.

Suppose that $X$ is a homochromatic set. Recall that $X$ is created
at time $t_0$, if the seed of $X$ is created at time step $t_0$. For
$t\geq t_0$, we use $X[t]$ to denoted the set of all nodes sharing
the same color as that created at time step $t_0$ at the end of time
step $t$. That is, we use $X[t]$ to denote a homochromatic set at
the end of time step $t$.

\bigskip

For Theorem~\ref{thm:Securityproperties} (4). Next, we prove the
small community phenomenon stated in
Theorem~\ref{thm:Securityproperties} (4).

Intuitively speaking, we will show that the homochromatic sets
created not too early or too late \footnote{From now on, whenever we
say that a homochromatic set appears at sometime, we mean that its
seed node appears at that time.} are good communities with high
probability. Then the conclusion follows from the fact that the
number of nodes in the remaining homochromatic sets only takes up a
$o(1)$ fraction.

We focus on the homochromatic sets created in time interval
$[T_3,T_4]$, where $T_3=\frac{n}{\log^{a+2} n},
T_4=\left(1-\frac{1}{\log^{(a-1)/2}n}\right)n$.

 Given a
homochromatic set $S$, we use $t_S$ to denote the time at which $S$
is created.

Let $S$ be a homochromatic set with $t_S\in[T_3,T_4]$, and let $s$
be the seed node of $S$. For any $t\geq t_S$, we use
$\partial(S)[t]$ to denote the set of edges from $S[t]$ to
$\overline{S[t]}$, the complement of $S[t]$. By
Definition~\ref{def:Securitymodel}, $\partial (S)[t]$ consists of
two types of edges:

\begin{enumerate}

  \item [(1)] The edges from the seed node of $S[t]$ to earlier nodes,
  i.e., the edges of the form $(t_S,j)$ for some $j$, and

   \item [(2)] The edges from the seed nodes created after time $t_S$ to nodes in $S[t]$

\end{enumerate}

By Definition~\ref{def:Securitymodel}, the number of edges of type
(1) above is at most $d$.

We only need to bound the number of the second type of edges. We
first make an estimation on the total degrees of nodes in $S[t]$ at
any given time $t>t_S$.

For each $t\geq t_S$, we use $D(S)[t]$ to denote the total degree of
nodes in $S[t]$ at the end of time step $t$ of Definition
\ref{def:Securitymodel}. We have the following lemma.

\begin{lemma} \label{lem:homodegree}
For any homochromatic set $S$ created at time $t_S\geq T_3$,
$D(S)[n]=O(\log^{a+1}n)$ holds with probability $1-o(1)$.
\end{lemma}

\begin{proof}
We only need to show that for any $t\geq T_3$, if $S$ is a
homochromatic set created at time step $t$, then
$D_n(S)[n]=O(\log^{a+1}n)$ holds with probability $1-o(n^{-1})$.
Without loss of generality, assume that $S$ is created at time step
$t_S=T_3$. The recurrence on $D(S)[t]$ can be written as
\begin{eqnarray*}
E[D(S)[t]\ |\ D(S)[t-1]] &=& D(S)[t-1]+\frac{1}{\log^a t} \left[
\frac{D(S)[t-1]}{2d(t-1)}+(d-1)\cdot\frac{1}{|C_{t-1}|} \right]\\
&& +\left(1-\frac{1}{\log^a t}\right)\cdot\frac{2d}{|C_{t-1}|}.
\end{eqnarray*}

We suppose again the event $\mathcal{E}$ that for all $t\geq
T_1=\log^{a+1}n$, $\frac{t}{2\log^a t}\leq |C_t|\leq
\frac{2t}{\log^a t}$, which almost surely happens by Lemma
\ref{lem:colorsize}. It holds also for $t\geq T_3$. On this
condition,

\begin{eqnarray}
E[D(S)[t]\ |\ D(S)[t-1],\mathcal{E}]&\leq& D(S)[t-1]
\left[1+\frac{1}{\log^a t}\frac{1}{2d(t-1)}\right] +\frac{2d}{|C_{t-1}|} \nonumber\\
&\leq& D(S)[t-1] \left[1+\frac{1}{\log^a t}\frac{1}{2d(t-1)}\right]
+\frac{4d\log^a t}{t}. \label{eqn:degree}
\end{eqnarray}
Then we use the submartingale concentration inequality (see
\cite{CL06}, Chapter 2, for information on martingales) to show that
$D(S)[t]$ is small with high probability.

Since
\begin{eqnarray*}
&&8d\log^{a+1}(t+1)-8d\left(1+\frac{1}{\log t}\frac{1}{2d(t-1)}\right)\cdot\log^{a+1} t\\
&\geq& 8d\log^a t\left(\log\frac{t+1}{t}\right)-\frac{8d\log^a t}{2d(t-1)}\\
&\geq& \frac{8d\log^a t}{t+1}-\frac{8d\log^a t}{2d(t-1)}\\
&\geq& \frac{4d\log^a t}{t},
\end{eqnarray*}
applying it to Inequality (\ref{eqn:degree}), we have
\begin{eqnarray*}
&&E[D(S)[t]\ |\ D(S)[t-1],\mathcal{E}]-8d\log^{a+1} (t+1)\\
&\leq& (1+\frac{1}{\log t}\frac{1}{2d(t-1)})(D(S)[t-1]-8d\log^{a+1}
t).
\end{eqnarray*}
For $t\geq T_3$, define $\theta_t=\Pi_{i=T_3+1}^{t} (1+\frac{1}{\log
i}\frac{1}{2d(i-1)})$ and $X[t]=\frac{D(S)[t]-8d\log^{a+1}
(t+1)}{\theta_t}$. Then
$$E[X[t]\ |\ X[t-1],\mathcal{E}]\leq X[t-1].$$
Note that
$$X[t]-E[X[t]\ |\ X[t-1],\mathcal{E}] = \frac{D(S)[t]-E[D(S)[t]\ |\ D(S)[t-1],E]}{\theta_t}\leq 2d,$$
Since
$$D(S)[t]-D(S)[t-1]\leq 2d,$$
we have
\begin{eqnarray*}
\Var[X[t]\ |\ X[t-1],\mathcal{E}] &=&
E[(X[t]-E(X[t]|X[t-1],\mathcal{E}))^2]\\
&=& \frac{1}{\theta_t^2}
E[(D(S)[t]-E(D(S)[t]\ |\ D(S)[t-1],\mathcal{E}))^2]\\
&\leq& \frac{1}{\theta_t^2}
E[(D(S)[t]-D(S)[t-1])^2|D(S)[t-1],\mathcal{E}]\\
&\leq& \frac{2d}{\theta_t^2}
E[D(S)[t]-D(S)[t-1]\ |\ D(S)[t-1],\mathcal{E}]\\
&\leq& \frac{2d}{\theta_t^2} \left[ \frac{4d\log^a
t}{t}+\frac{1}{\log^a t}\cdot\frac{D(S)[t-1]}{2d(t-1)} \right]\\
&=& \frac{8d^2\log^a t}{t\theta_t^2}+\frac{1}{(t-1)\theta_t\log^a
t}\cdot\frac{D(S)[t-1]}{\theta_t}\\
&\leq& \frac{8d^2\log^a t}{t\theta_t^2}+\frac{8d\log^{a+1}
t}{(t-1)\theta_t^2\log^a t}+\frac{X[t-1]}{(t-1)\theta_t\log^a
t}\\
&\leq& \frac{9d^2\log^a
t}{t\theta_t^2}+\frac{X[t-1]}{2d(t-1)\theta_t\log^a t}.
\end{eqnarray*}
Note that $\theta_t$ can be bounded as
\begin{eqnarray*}
\theta_t&\sim& e^{\sum_{i=T_3+1}^t\frac{1}{2d(i-1)\log i}}\in
[(\frac{t}{T_3})^{\frac{1}{2d\log n}},
(\frac{t}{T_3})^{\frac{1}{2d\log T_3}}].
\end{eqnarray*}
Then
\begin{eqnarray*}
\sum_{i=T_3+1}^t\frac{9d^2\log^a i}{i\theta_i^2}\leq 9d^2\log^a n
\int_{T_3}^t\frac{1}{i}\cdot\left(\frac{T_3}{i}\right)^{\frac{1}{d\log
n}}di \leq 9d^2\log^a n \cdot \log n=9d^2\log^{a+1} n,
\end{eqnarray*}
and
\begin{eqnarray*}
\sum_{i=T_3+1}^t\frac{1}{2d(i-1)\theta_i\log^a i}\leq
\frac{1}{d\log^a T_3}\int_{T_3}^t\frac{T_3^{\frac{1}{2d\log
n}}}{i\cdot i^{\frac{1}{2d\log n}}}di \leq \frac{\log n}{d\log^a
T_3}.
\end{eqnarray*}

Here we can safely assume that $X[t]$ is non-negative, which means
that $D(S)[t]\geq 8\log^{a+1}(t+1)$, because otherwise, the
conclusion follows immediately. Let $\lambda=10\log^{a+1}n$. By the
submartingale inequality (\cite{CL06}, Theorem 2.40),
\begin{eqnarray*}
&&\Pr[X[t]=\omega(\log^{a+1}n)]\leq\Pr[X[t]\geq X[T_3]+\lambda]\\
&\leq& \exp(-\frac{\lambda^2}{2(9d^3\log^{a+1} n+  10\log^{a+1}
n+d\lambda/3)})+O(n^{-2})=O(n^{-2}).
\end{eqnarray*}
This implies that $D(S)[n]=O(\log^{a+1}n)$ holds with probability
$1-O(n^{-2})$.
\end{proof}

Let $S$ be a homochromatic set created at some time $t_S<T_4$. Let
$s$ be the seed node of $S$. We consider the edges from seed nodes
created after time step $t_S$ to nodes in $S$. For $t>t_S$, if a
seed node, $v$ say, is created at time step $t$, then there are two
types of edges from $v$ to nodes in $S[t-1]$, they are:

\begin{enumerate}
\item [(1)] (First type edges) An edge $(v,u)$ for some $u\in S[t-1]$ created in step
(3) (b) of Definition \ref{def:Securitymodel}.

We call the edges created in (1) are the {\it first type edges}.

\item [(2)] (Second type edges) Some edges $(v,s)$ for  the seed node $s\in S[t-1]$
created by step (3) (c) of Definition \ref{def:Securitymodel}.

We call the edges created in (2) above the {\it second type edges}.

\end{enumerate}

 We will bound the numbers of these two types of edges,
respectively.

By a similar proof to that in Lemma~\ref{lem:setsize} (1), we are
able to show that, with probability $1-o(1)$, $S=S[n]$ has a size
$\Omega(\log^{\frac{a+1}{2}} n)$, and so a volume
$\Omega(\log^{\frac{a+1}{2}} n)$. We suppose the event, denoted by
$\mathcal{F}$, that for any $t\geq T_S$, $D(S)[t]=O(\log^{a+1}n)$,
which holds with probability $1-o(1)$ by Lemma \ref{lem:homodegree}.
For each $t\geq T_S$, we define a $0,1$ random indicator variable
$X_t$ which indicates the event that the first type edge connects to
$S$ at time $t$ and satisfies
$$\Pr[X_t=1|\mathcal{F}]=\frac{1}{\log^a t}\frac{D(S)[t-1]}{2d(t-1)}\leq \frac{\log^{1+\epsilon} n}{2d(t-1)},$$
for arbitrarily small positive $\epsilon$, i.e.,
$0<\epsilon<\frac{a-1}{4}$. Then
$$E[\sum_{t=t_S}^n X_t]\leq \log^{1+\epsilon}n\sum_{t=t_S}^n\frac{1}{2d(t-1)}\leq (\log^{1+\epsilon}n)(\log\log n).$$
By the Chernoff bound,
$$\Pr[\sum_{t=t_S}^n X_t\geq 2(\log^{1+\epsilon}n)(\log\log n)]\leq n^{-2}.$$
That is, with probability at least $1-n^{-2}$, the total number of
first type edges is upper bounded by $2(\log^{1+\epsilon}n)(\log\log
n)$.

For the second type of edges, conditioned on the event
$\mathcal{E}$, this number is expected to be at most
$$\sum\limits_{t=T_3}^n \frac{1}{\log^a t}\cdot\frac{1}{|C_t|}\cdot(d-1)
\leq O(\sum\limits_{t=T_3}^n \frac{1}{\log^a t}\cdot\frac{2\log^a
t}{t})=O(\log \log n).$$ So by the Chernoff bound, with probability
$1-o(1)$, the number of second type of edges is upper bounded by
$O(\log n)$.

Hence, with probability $1-o(1)$, the conductance of $S$ is
$$\Phi (S)=O\left(\frac{2(\log^{1+\epsilon}n)(\log\log n)+\log n}{\log^{(a+1)/2} n}\right)
\leq O\left(\log^{-\frac{a-1}{4}}n\right)\leq
O\left(|S|^{-\frac{a-1}{4(a+1)}}\right).$$

The total number of nodes belonging to the homochromatic sets which
appear before time $T_3$ or after time $T_4$ is at most
$\log^{a+1}n\cdot\frac{n}{\log^{a+2}n}+\frac{n}{\log^{(a-1)/2}
n}=o(n)$ for any constant $a>1$. Therefore, $1-o(1)$ fraction of
nodes of $G$ belongs to a subset $W$ of nodes, which has a size
bounded by $O(\log^{a+1} n)$ and a conductance bounded by
$O\left(|W|^{-\frac{a-1}{4(a+1)}}\right)$. This proves Theorem
\ref{thm:Securityproperties} (4).

This completes the proof of Theorem \ref{thm:Securityproperties}.

\section{Probabilistic and Combinatorial Principle} \label{sec:pcomb}

Theorem \ref{thm:Securityproperties} provides the necessary
structural properties for proving Theorems
\ref{thm:cascadeonSecurity}, and \ref{thm:rancascadeonSecurity}. In
this section, we prove the necessary probabilistic and combinatorial
principles for the proofs of the security theorems, that is, Theorem
\ref{thm:length}, and Theorem \ref{thm:injury}.

\subsection{Degree Priority Theorem} \label{sec:comb}

In this subsection, we prove Theorem \ref{thm:length}.

\begin{proof} (Proof of Theorem~\ref{thm:length}) For (1). To bound the expected
length of degrees for all nodes, it suffices
to bound the length of degrees of seed nodes. Let $v$ be a seed node
created at time $t_0$.

By Lemma~\ref{lem:colorsize}, for each $t$, $|C_t|$ is expected to
be $\Theta(\frac{t}{\log^a t})$. Thus the expected number of seed
nodes created after time $t_0$ and linked to $v$ is at most
$d\cdot\frac{1}{\log^a t}\cdot\frac{1}{|C_t|}=O(\frac{1}{t})$. This
shows that
$$E[l(v)]=O(\sum\limits_{t=1}^n \frac{1}{t})=O(\log n).$$

(1) follows.

\smallskip

For (2), (3) and (4). We prove (2) - (4) together by considering two
cases:

Case 1. $v$ is a non-seed node.

Suppose that $v$ is created at time step $t_0$. We use $D(v)$ to
denote the degree of $v$ contributed by nodes of the same color as
$v$, and $F(v)$ to denote the maximal degree of $v$ contributed by
nodes that share the same color other than the color of $v$. By (4)
of Definition~\ref{def:Securitymodel}, $D(v)[t_0]=d$, and
$F(v)[t_0]=0$.

For $t+1>t_0$, let $u$ be the node created at time step $t+1$. If
$u$ is a seed node, then by (3) of
Definition~\ref{def:Securitymodel}, we have that $D(v)[t+1]=D(v)[t]$
and $F(v)[t+1]\leq \max\{F(v)[t], 1\}$. If $u$ is a non-seed node,
then either $u$ has the same color as that of $v$, or $u$ chooses an
 old color different from that of $v$, in either case, we have that
$D(v)[t+1]\geq D(v)[t]$ and $F(v)[t+1]=F(v)[t]$.

Therefore, we have that the first degree of $v$, $d_1(v)$ is always
contributed by the neighbors of $v$ that share the same color as
$v$, that is, $D(v)$, and that the second degree $d_2(v)\leq 1$.

Case 2. $v$ is a seed node.

Let $v$ be a node created at time step $t_0$. We use $F(v)[t]$ to
denote the largest number of homochromatic neighbors having
different color from $v$ at the end of time step $t$.

By step (3) of Definition~\ref{def:Securitymodel}, $F(v)[t_0]\leq
d$. For every $t\geq t_0$, We consider time step $t+1$. Let $u$ be
the node created at time step $t+1$. If $u$ is a seed node, then by
(3) of Definition~\ref{def:Securitymodel}, we have that
$F(v)[t+1]\leq \max\{F(v)[t], d\}$. If $u$ is a non-seed node, then
by (4) of Definition~\ref{def:Securitymodel}, $F(v)[t+1]=F(v)[t]$.

Therefore, we have that $F(v)[n]\leq d =O(1)$.

Next we consider the degree of $v$ contributed by the neighbors of
the same color as $v$. Note that a seed node has a degree at least
$d$ contributed by local edges, unless the homochromatic set of the
seed node is too small. This kind of seed nodes is likely to be
created too late. We choose an appreciate time stamp $T$ and show
that there are only a negligible number of seed nodes born after $T$
and all the seed nodes born before $T+1$ are contained in
homochromatic sets of non-negligible size and thus have a large
degree contributed by local edges.

Here we choose the time step $T=T_4$, defined in Subsection
\ref{sub3.1(3)}.

By the proof of Lemma \ref{lem:setsize}, the homochromatic sets
created at time step $\leq T_4$ has size at least $\Omega
(\log^{\frac{a+1}{2}} n)$ with probability $1-o(1)$. The next lemma
guarantees that a seed node of a homochromatic set of size $\Omega
(\log^{\frac{a+1}{2}} n)$ has degree $\Omega (\log^{\frac{a+1}{4}}
n)$  contributed by local edges.

By Definition \ref{def:Securitymodel} (4), the induced subgraph of a
homochromatic set basically follows the PA scheme, so it suffices to
prove a result for networks of the PA model.

\begin{lemma} \label{lem:degreeexpectation} Suppose that $G$ is a network
generated from the preferential attachment model. Let $v_i$ be the
$i$-the vertex in $G$. Then we have that the degree of $v_i$ is
expected to be $\sqrt{\frac{n}{i}}\cdot d$.
\end{lemma}

\begin{proof}
Let $s_i$ be the expected degree of $v_i$. Fix $i$, and for $j\geq
i$, let $a_i(j)$ be the expected degree contributed by $v_j$ to
$v_i$ and $T_i(j)$ be the expected degree of $v_i$ at the end of
step $j$. So for each $i$, $a_i(i)=T_i(i)=d$, $T_i(n)=s_i$ and
$T_i(j)=\sum_{k=i}^j a_i(k)$. Note that the volume of the whole
graph at step $j$ is $2dj$. For $j\geq i$,
$a_i(j+1)=\frac{T_i(j)}{2dj}\cdot d=\frac{T_i(j)}{2j}$, and hence
$T_i(j+1)=T_i(j)+\frac{T_i(j)}{2j}$. By this recurrence equation, we
have
$$T_i(n)=\prod\limits_{j=i}^{n-1}
(1+\frac{1}{2j})\cdot T_i(i).$$ Define a function
$f(m)=\prod\limits_{j=1}^{m-1} (1+\frac{1}{2j})$. So
$T_i(n)=\frac{f(n)}{f(i)}\cdot d$. Since
$f(n)=\frac{(2n-1)!}{2^{2(n-1)}[(n-1)!]^2}$, by the Stirling
formula, when $n$ is large enough,
$f(n)=\frac{2}{\sqrt{\pi}}\cdot\sqrt{n}$. Thus,
$s_i=T_i(n)=\sqrt{\frac{n}{i}}\cdot d$.
\end{proof}

So by step (3) of Definition~\ref{def:Securitymodel}, with
probability $1-o(1)$, a homochromatic set of size at least $\Omega
(\log^{\frac{a+1}{2}} n)$ has a seed node of degree at least $\Omega
(\log^{\frac{a+1}{4}} n)$  contributed by local edges. So the seed
nodes created at time step $\leq T_4$ have their first degrees
contributed by local edges with probability $1-o(1)$.

By the proof in Sunsection \ref{sub3.1(3)}, the number of seed nodes
created after time step $T_4$ is negligible.

 Therefore with probability $1-o(1)$, a
randomly picked seed node has its first degree contributed by its
neighbors sharing the same color as the seed node.

All (2), (3) and (4) follow from Cases 1 and 2.

This completes the proof of Theorem~\ref{thm:length}.
\end{proof}

\subsection{Infection-Inclusion Theorem} \label{subsec:injury}

In this subsection, we prove Theorem \ref{thm:injury}. At first, we
give a basic definition of communities, targeted communities, and
infected communities.

\begin{definition} \label{def:communityinjury} Let $G$ be a network
constructed from the security model.

\begin{enumerate}
\item [(1)] A community of $G$ is the induced subgraph of a
homochromatic set of $G$.
\item [(2)] We say that a community, $G_X$ say, is created at time
step $t$, if the seed node of $G_X$ is created at time step $t$.
\item [(3)] We say that a community, $G_X$ say, is targeted, if
there is a node in $X$ which is targeted by an attack, and
non-targeted, otherwise.
\item [(4)] We say that a community $G_X$ is infected, if there is a
node in $X$ which has been either targeted or infected, and
non-infected, otherwise.

\end{enumerate}

\end{definition}

\begin{proof} (Proof of Theorem~\ref{thm:injury}) For (1).  We consider two cases:

For (i). The infection of $G_Y$ from a non-seed node $x_1$ in $G_X$.

By Definition~\ref{def:Securitymodel}, there is no edges between
non-seed nodes in $G_X$ and non-seed nodes in $G_Y$, and there is no
edge between the seed node of $G_X$ and non-seed nodes in $G_Y$.

Therefore, there is no injury from $G_X$ to any non-seed node in
$G_Y$. Hence the only possible node in $G_Y$ which may be injured by
$G_X$ is the seed node $y_0$ of $G_Y$. (i) follows.

For (ii). The injury of the seed node in $G_Y$ from $G_X$.

 By
Theorem~\ref{thm:length}, the number of neighbors of the seed node
$y_0$ (of $G_Y$) in $G_X$ is less than or equal to the second degree
of $y_0$, which is at most a constant.

For (2). Suppose that $x_1$ and $y_1$ are non-seed nodes in $X$ and
$Y$ respectively.

For (i). The injury of $G_Y$ from the non-seed node $x_1$.

This fails to occur since at the stage at which $x_1$ is created, it
links to nodes only in $G_X$.

For (ii). The injury of the seed node $y_0$ of $G_Y$ from the whole
community $G_X$.

In this subcase, the possible neighbors of $y_0$ in $G_X$ is  only
the seed $x_0$ of $G_X$, and $y_0$ is a seed node of $G_Y$.
Therefore the injury of $y_0$ from $G_X$ is bounded by $1$.

For (iii). The injury of a non-seed node $y_1$ from $G_X$.

The same as that in (i) and (ii) above, the only possible neighbors
of $y_1$ in $G_X$ is the seed node $x_0$ of $G_X$. In this case, by
Definition~\ref{def:Securitymodel}, the only possibility that
 there is a link between $x_0$ and and a non-seed node $y$ of $G_Y$ is that $y$ is the unique
 node chosen by the preferential attachment scheme in step (3) (b)
 of Definition~\ref{def:Securitymodel} at the time step at
 which $x_0$ is created.

(3) and (4) follow from (1) and (2).

 This completes the proof of Theorem~\ref{thm:injury}.
\end{proof}

\section{Security Theorems of the Security Model}
\label{sec:cascadeonSecurity}

In this section, we will prove the security theorems of the security
model, i.e., Theorems \ref{thm:cascadeonSecurity} and
\ref{thm:rancascadeonSecurity}, by applying the fundamental theorem,
i.e., Theorem \ref{thm:Securityproperties}, and the probabilistic
and combinatorial principles in Theorems~\ref{thm:length}, and
\ref{thm:injury}.

\subsection{Infection Priority Tree} \label{subsec:basiclemma}

In this subsection, we propose the notion of infection priority tree
of a network and develop the key lemmas to the proofs of Theorems
\ref{thm:cascadeonSecurity} and \ref{thm:rancascadeonSecurity}, by
using Theorems~\ref{thm:length}, and \ref{thm:injury}.

At first, we have that

\begin{lemma} \label{lem:injurylemma2} (Infection Lemma)
For any communities $G_X$ and $G_Y$, the injury of $G_Y$ from the
whole community $G_X$ satisfies:

\begin{enumerate}
\item For the seed node $y_0$ of $G_Y$, the injury of $y_0$ from
$G_X$ is bounded by $O(1)$.

\item For a non-seed node $y\in Y$, $G_X$ injures $y$, only if the
following occurs:

\begin{itemize}
\item $y$ is injured only by the seed node $x_0$ of $G_X$,
\item $y$ is created before the creation of the seed $x_0$ of $G_X$, and
\item At the time step at which $x_0$ is created, (3) (b) of
Definition~\ref{def:Securitymodel} occurs, which creates an edge
$(x_0,y)$.

\end{itemize}

\end{enumerate}
\end{lemma}
\begin{proof} By Theorem~\ref{thm:injury}.
\end{proof}

By Theorem~\ref{thm:Securityproperties} (2) (i), every community has
size bounded by $O(\log^{a+1} n)$, we can safely assume the
following:

\begin{definition} \label{def:conv} (Convention)
For any community $G_X$, if there is a node $x\in X$ is either
targeted or infected, then all the nodes in $X$ have been infected.
\end{definition}

By Definition~\ref{def:conv}, we consider only the infections among
different communities. By Lemma~\ref{lem:injurylemma2}, we only
consider two types of injuries among two communities.

\begin{definition} \label{injurytype} (Injury Type) We define:

\begin{enumerate}
\item (First type) The first type of injury is the injury of a seed
node.

\item (Second type) The second type is an injury following an edge created by (3)
(b) of Definition~\ref{def:Securitymodel}.
\end{enumerate}
\end{definition}

To deal with the first type injury, we introduce the notion of
strong communities.

\begin{definition} \label{def:strong} Given a homochromatic set $X$,
suppose that $x_0$ is the seed node of $X$, and that $G_X$ is the
community induced by $X$.

We say that $G_X$ is a strong community, if the seed node $x_0\in X$
will never be infected, unless there is a node $x\in X$ which has
already been infected. Otherwise, we say that $G_X$ is a vulnerable
community.

\end{definition}

By Theorem~\ref{thm:length}, for every seed node $x$ of a community
$G_X$, the length of degrees of $v$ is bounded by $O(\log n)$, and
the second degree of $v$ is bounded by $O(1)$, therefore the injury
of the seed node $x$ from the collection of all communities other
than $G_X$ itself can be bounded by $O(\log n)$.  This allows us to
show that for any set of attacks of poly logarithmic sizes, almost
surely, there is a huge number of strong communities.

By Lemma~\ref{lem:injurylemma2}, the injury among strong communities
is the second type. To analyze the infections among the strong
communities, we define {\it the infection priority tree } $T$ of $G$
by modulo the small communities from the network.

\begin{definition} \label{def:reduction} (Defining infection priority tree $T$) Let $G$ be a network
constructed by Definition~\ref{def:Securitymodel}. We define the
infection priority tree $T$ to be a directed graph as follows:

\begin{enumerate}
\item Let $H$ be the graph obtained from $G$ by deleting all the edges constructed by (3) (c) of
Definition~\ref{def:Securitymodel}, keeping the directions in $G$.
\item Let $T$ be the directed graph obtained from $H$ by merging
each of the homochromatic sets into a single node.

\end{enumerate}

\end{definition}

Then we have that

\begin{lemma} \label{lemma:direction}
Any infection from a strong community to a strong community must be
triggered by a directed edge in the infection priority tree $T$.

\end{lemma}

\begin{proof}
By Definition~\ref{def:reduction}, Definition~\ref{def:strong}, and
Theorem~\ref{thm:injury}.
\end{proof}

Lemma~\ref{lemma:direction} shows that the cascading behavior in the
infection priority tree $T$ is always directed from a seed node to
an old non-seed node created in (3) (b) of
Definition~\ref{def:Securitymodel}.

Now the key to our proofs is that cascading procedure in $T$ must
terminate shortly, that is, after $O(\log n)$ many steps.

\begin{lemma} \label{lemma:height} With probability $1- o(1)$, the
following hold:
\begin{enumerate}
\item The infection priority tree $T$ is a directed tree.
\item The height of the infection priority tree $T$ is $O(\log n)$.
\end{enumerate}

\end{lemma}

\begin{proof} By Definition~\ref{def:Securitymodel} and
Definition~\ref{def:reduction}, $T$ can be regarded as a graph
constructed by a preferential attachment scheme with $d=1$ such that
whenever a new node is created, it links to a node chosen with
probability proportional to the weights of nodes, at the same time,
the weights of nodes are increasing uniformly and randomly.
Precisely, we restate the construction of $T$ as follows:

\begin{enumerate}
\item [(i)] Take $H_2$ to be a graph with two nodes $1, 2$, one
 directed edge $(2,1)$ such that each node has a weight $w(i)=d$ for $i=1, 2$.

 For $i+1>2$, let $p_i=1/(\log i)^a$, and let $H_i$ be
the graph constructed at the end of time step $i$.

\item [(ii)] With probability $p_i$, we create a new node, $v$ say,
in which case,

\begin{enumerate}

\item let $u_0$ be a node chosen with probability proportional to the
weights of nodes in $H_i$, create a directed edge $(v,u_0)$,

\item let $u_1, u_2,\cdots, u_{d-1}$ be nodes chosen randomly and
uniformly in $H_i$,

\item for each $j=0,1,\cdots, d-1$, set $w(u)\leftarrow {\rm old}\
w(u)+1$, and

\item  set $w(v)[i+1]=d$.

\end{enumerate}

\item [(iii)] Otherwise, then choose randomly and uniformly a node,
$u$ say, in $H_i$, set $w(u)[i+1]=w(u)[i]+2d$.

\end{enumerate}

Then $T$ is the directed graph obtained from $H$ by ignoring the
weights of nodes.

For (1). Clearly, it is true that $T$ is a tree, because whenever
one new node is created, there is only one new edge is added, and
the graph is connected. (1) holds.

For (2). By definition of $T$, the height of $T$ is between a graph
of the preferential attachment model with $d=1$ and a uniform
recursive tree of the same number of nodes. By
Lemma~\ref{lem:resursivetrees}, with probability $1-o(1)$, a uniform
recursive tree of nodes $n$ has height bounded by $O(\log n)$. By
construction above, $T$ has height stochastically dominated by that
of a uniform recursive tree of the same number of nodes. Therefore,
with probability $1-o(1)$, the height of $T$ is bounded by $O(\log
n)$. (2) holds.
\end{proof}

\smallskip

By Lemmas~\ref{lemma:direction} and \ref{lemma:height}, $T$ exactly
captures the cascading behaviors among strong communities, which is
the key to our proofs.

Now we know that the proofs of both
Theorem~\ref{thm:cascadeonSecurity} and
Theorem~\ref{thm:rancascadeonSecurity} consist of the following
steps:

\begin{enumerate}

\item To prove that for any attack of poly logarithmic size, almost surely, there
is a huge number of strong communities.

\item Any infection among the strong communities must be triggered by
an edge in the infection priority tree of $G$, which goes at most
$O(\log n)$ many steps, by Lemma~\ref{lemma:height}.

\end{enumerate}

(2) has been guaranteed by Lemma~\ref{lemma:direction} and
Lemma~\ref{lemma:height}. So the main issue for the proofs of
Theorem~\ref{thm:cascadeonSecurity} and
Theorem~\ref{thm:rancascadeonSecurity} is actually step (1) above,
which will be given in Subsections~\ref{proof:uthm} and
\ref{subsec:randomsecuritytheorem}.

\subsection{Uniform Threshold Security Theorem} \label{proof:uthm}

In this subsection, we prove Theorem~\ref{thm:cascadeonSecurity}.

Let $G$ be a network constructed by the security model. Consider a
deliberate attack by targeting an initial set $S$ of size poly$(\log
n)$. Note that the size of $S$, poly$(\log n)$, is much smaller than
the number of communities, i.e., $\Theta(n/\log^a n)$, by (1) (i) of
Theorem~\ref{thm:Securityproperties}.

\begin{proof} (Proof of Theorem \ref{thm:cascadeonSecurity})
Set time $T_0=(1-\delta)n$, where $\delta=\log^{-b_0} n$, where
$b_0$ will be determined later. We will show that with high
probability, all the communities created before time step $T_0$ are
large and thus strong.

\begin{lemma} \label{lem:homosetsizebeforeT0} Let $2<b_1<a-b_0$. Then
with probability $1-o(1)$, every homochromatic set created before
time step $T_0$ has a size $\Omega(\log^{b_1}n)$.
\end{lemma}

\begin{proof}
It is sufficient to show that, with probability $1-n^{-1}$, for
every homochromatic set $S_\kappa$ created before $T_0$,
$S_{\kappa}$ has a size $\Omega(\log^{b_1}n)$.

Suppose that $S_{\kappa}$ is the set with color $\kappa$, and that
it is created at time step $t_0\leq T_0$ for some $t_0$. For any
$t\geq t_0$, define an indicator random variable $Y_t$ to be the
event that the node created at time step $t$ chooses color $\kappa$.

Define $\{Z_t\}$ to be the independent Bernoulli trails such that
$$\Pr[Z_t=1]=\left(1-\frac{1}{\log^a (1-\delta)n}\right)\frac{\log^a
t}{2t}.$$

Conditioned on the event $\mathcal{E}$ in
Definition~\ref{def:event}, we have that random variable
$Y:=\sum_{t=t'}^n Y_t$ stochastically dominates $Z:=\sum_{t=t'}^n
Z_t$ for any $t'\leq T_0$.

By definition, $Z$ has an expectation
$$E[Z]\geq\left(1-\frac{1}{\log^a (1-\delta)n}\right)
\sum_{t=T_0+1}^n\frac{\log^a t}{2t}\geq \frac{\delta n}{2n}\log^{a}
(1-\delta)n=\Omega(\log^{a-b_0} n).$$ Since $2<b_1<a-b_0$, by the
Chernoff bound,
$$\Pr\left[Z=O(\log^{b_1}n)\right]\leq n^{-1}.$$
Therefore, with probability $1-n^{-1}$, the size of $S_\kappa$ is at
least $Y=\Omega(\log^{b_1}n)$.
\end{proof}

Secondly, we show that every seed node created before $T_0$ probably
has a large degree.

\begin{lemma} \label{lem:seeddegree}
With probability $1-o(1)$, every seed node created before time step
$T_0$ has degree at least $\Omega(\log^{b_1/2} n)$.
\end{lemma}

\begin{proof} Let $v$ be a seed node created at a time step $\leq
T_0$. Suppose that $v$ has color $\kappa$. Let $S$ be the set of all
nodes sharing color $\kappa$. Then the community $G_S$ is the
induced subgraph of $S$ in $G$. The degree of the seed node $v$ in
$G$ is contributed by both local edges and global edges. By the
construction, $G_S$ truthfully follows a power law, by Lemma
\ref{lem:degreeexpectation}, the degree of $v$ contributed by local
edges is expected at least $\sqrt{|S|}$. By Lemma
\ref{lem:homosetsizebeforeT0}, with probability $1-o(1)$, each $S$
has a size $\Omega(\log^{b_1}n)$. The degree of $v$ has an expected
degree at least $\Omega(\log^{b_1/2} n)$. Since $b_1>2$, by the
Chernoff bound, with probability $1-o(n^{-1})$, $v$'s degree is at
least $\Omega(\log^{b_1/2} n)$. The lemma follows immediately by the
union bound.
\end{proof}

Now we are able to estimate the number of strong communities.

\begin{lemma} \label{lem:strongbeforeT0}
Let $b_0=2+\epsilon$ and $b_1=a-b_0-\frac{\epsilon}{2}$, where
$\epsilon$ is that defined in Theorem \ref{thm:cascadeonSecurity}.
With probability $1-o(1)$, all the communities created before time
$T_0$ are strong.
\end{lemma}

\begin{proof} By Theorem~\ref{thm:length},
the length of degrees of a seed node is bounded by $O(\log n)$, and
the second degree of a seed node is bounded by $O(1)$. By Chernoff
bound, we have that, with probability $1-o(1)$, for every seed node
$v$, the degree of $v$ contributed by global edges is bounded by
$O(\log n)$. By Lemma \ref{lem:seeddegree}, almost surely, for each
seed node $v$, the fraction of $v$'s degree contributed by global
edges is less than or equal to $O(\log^{1-b_1/2}n)$. Recall that the
threshold parameter $\phi=\Omega\left(\frac{1}{\log^b n}\right)$ for
$b=\frac{a}{2}-2-\epsilon$ for arbitrary $\epsilon>0$. By the
choices of $b_0$ and $b_1$,
$1-b_1/2=-\left(\frac{a}{2}-2-\frac{3\epsilon}{4}\right)<-b$. The
lemma follows.
\end{proof}

For the total number of vulnerable communities, we have

\begin{lemma} \label{lem:vulnumber}
Let $b_2=a+b_0$. With probability $1-o(1)$, the number of vulnerable
communities is at most $\frac{2n}{\log^{b_2}n}$.
\end{lemma}

\begin{proof}
By Lemma \ref{lem:strongbeforeT0}, we only need to bound the number
of communities created after time step $T_0$. Since at time step
$t$, a new color is created with probability $p_t=\log^{-a}t$, the
number of colors created after time step $T_0$, denoted by $N_{\rm
vul}$ is expected to be
$$E[N_{\rm vul}]=\sum\limits_{t=T_0+1}^n \frac{1}{\log^a t}.$$
When $n$ is large enough, by a simple integral computation,
$E[N_{\rm vul}]$ is upper bounded by $\frac{3\delta n}{2\log^a n}$.
By the Chernoff bound, with probability $1-o(1)$, $N_{\rm vul}$ is
at most $\frac{2\delta n}{\log^a n}=\frac{2n}{\log^{b_2} n}$. The
lemma follows.
\end{proof}

Now we are ready for the proof of Theorem
\ref{thm:cascadeonSecurity}.

Suppose that $S$ is the initially targeted set of size $\lceil\log^c
n\rceil$. Choose $b_0=2+\epsilon$, $b_1=a-b_0-\frac{\epsilon}{2}$
and $b_2=a+b_0$.

By Lemma \ref{lem:vulnumber}, with probability $1-o(1)$, the number
of  vulnerable communities is at most $\frac{2n}{\log^{b_2}n}$. By
Lemma~\ref{lemma:height}, the height of infection priority tree $T$
is $h=O(\log n)$. By Lemma~\ref{lemma:direction}, infections among
strong communities must be triggered by an edge in the infection
priority tree $T$. Therefore the number of infected communities by
attacks on $S$ is at most

$$\left(|S|+\frac{2n}{\log^{b_2}n}\right)\cdot h=O\left(\left(\lceil\log^c
n\rceil+\frac{2n}{\log^{b_2}n}\right)\cdot \log n\right).$$

By Theorem~\ref{thm:Securityproperties} (1), with probability $1-
o(1)$, the largest community has a size $O(\log^{a+1}n)$. So the
number of infected nodes in $G$ by attacks on $S$ is at most

$$O\left(\left(\lceil\log^c
n\rceil+\frac{2n}{\log^{b_2}n}\right)\cdot \log n\cdot
\log^{a+1}n\right)=o(n).$$

This completes the proof of Theorem \ref{thm:cascadeonSecurity}.

\end{proof}

The proof of Theorem \ref{thm:cascadeonSecurity} is essentially a
methodology of community analysis of networks of the security model.
The key ideas of the methodology are those in
Theorems~\ref{thm:Securityproperties}, \ref{thm:length}, and
Theorem~\ref{thm:injury}, Definition~\ref{def:strong},
Definition~\ref{def:reduction}, Lemma~\ref{lemma:direction}, and
Lemma~\ref{lemma:height}.

The method allows us to divide all the communities into two classes,
the first is the strong communities, and the second is the
vulnerable ones. The two types of communities are distinguished by a
time step $T_0$. This time stamp $T_0$ is determined by both
parameter $\delta$, and essentially by the power $b$. Then we show
that communities created before time step $T_0$ are strong, and that
the number of communities created after time step $T_0$ is small.

Theorem \ref{thm:cascadeonSecurity} shows that the power law
distribution in Theorem~\ref{thm:Securityproperties}, is never an
obstacle for security of networks. Our proof of the security theorem
show that the community structure of the networks isolates the
vulnerable nodes in a large number of small communities, that the
homogeneity and randomness among the seed nodes or ``hubs" guarantee
that most communities are strong, and that the infection priority
tree ensures that the cascading procedure among strong communities
cannot be long.

\subsection{Random Threshold Security Theorem}
\label{subsec:randomsecuritytheorem}

In this subsection, we prove Theorem~\ref{thm:rancascadeonSecurity}.
The proof has the same framework as before. By Lemmas
\ref{lemma:direction}, and \ref{lemma:height}, infections among
strong communities must be triggered by edges in the infection
priority tree $T$, and infections in $T$ are directed, and terminate
by $O(\log n)$ many steps.

Therefore, the only issue is to prove that the number of vulnerable
communities is small.

\begin{proof}

Let $T_0=(1-\delta)n$, where $\delta=100\log^{-b_0} n$ and $b_0$ to
be determined later. Let $T_0'=n/100$.

By a similar proof to that of Lemma \ref{lem:homosetsizebeforeT0},
for every $b_1\in(1,a-b_0]$, we have that  with probability
$1-o(1)$, the following hold:

\begin{itemize}

\item \ Every community created at a time step $t\leq T_0'$ has a size
$\Omega(\log^a n)$, and

\item  \ Every community created at a time step $t\in [T_0', T_0]$ has a size
$\Omega(\log^{b_1} n)$.

\end{itemize}

By the proof of Lemma \ref{lem:seeddegree}, we have that with
probability $1-o(1)$,

\begin{enumerate}

\item A seed node created at a time step $t\leq T_0'$ has degree $\Omega(\log^{a/2}n)$, and

\item A seed node created at a time step $t\in [T_0', T_0]$ has
degree $\Omega(\log^{b_1/2}n)$.

\end{enumerate}

Then we show that the number of vulnerable communities created
before time step $T_0$ is small.

\begin{lemma}
Let $b_0=\frac{a}{2}-1$ and $b_1=\frac{a}{2}+1$. With probability
$1-o(1)$, there are only $O\left(\frac{n}{\log^{a+(b_1/2)}
n}\right)$ communities created before time step $T_0$ that are
vulnerable.
\end{lemma}

\begin{proof}

By the Chernoff bound, with probability $1-o(1)$:

(i) By Theorem~\ref{thm:length}, every seed node created before time
step $T_0'$ has a degree at most  $O(\log n)$ contributed by global
edges, and

(ii) All but $O(\log n)$ seed nodes created in time interval
$[T_0',T_0]$ have a degree $O(1)$ contributed by global edges.

Note that the threshold of each node is chosen randomly and
uniformly. Then the communities that are created in these two time
slots and satisfy the above conditions are vulnerable with
probability $O(\log^{1-(a/2)}n)$ and $O(\log^{-b_1/2}n)$,
respectively.

By Theorem~\ref{thm:Securityproperties} (1), with probability
$1-o(1)$, there are at most $O\left( \frac{2n}{\log^a n} \right)$
communities. By the choice of $a>6$,  $-b_1/2>1-(a/2)$ holds.
Therefore, the expected number of vulnerable communities created
before time step $T_0$ is $O\left(\frac{n}{\log^{a+(b_1/2)}
n}\right)$.

 Noting the independence of choice of threshold for each
node, by using the Chernoff bound again, the lemma follows.
\end{proof}

By the proof of Lemma \ref{lem:vulnumber}, there are only
$O\left(\frac{n}{\log^{a+b_0}n}\right)$ communities born after
$T_0$. So the total number of vulnerable communities in $G$ is
$O\left(\frac{n}{\log^{a+b_0}n}+\frac{n}{\log^{a+(b_1/2)}n}\right)
=O\left(\frac{n}{\log^{a+b_0}n}\right)$.

Consider the infection priority tree $T$ again. For any initial
targeted set $S$ of size $\lceil\log^c n\rceil$, the size of ${\rm
inf}_H^{\rm U}(S)$ is at most
$$O\left(\left(\lceil\log^c
n\rceil+\frac{2n}{\log^{a+b_0}n}\right)\cdot \log n\cdot
\log^{a+1}n\right)=o(n).$$

This completes the proof of Theorem \ref{thm:rancascadeonSecurity}.
\end{proof}

\subsection{Framework for Security Analysis}

Theorems~\ref{thm:cascadeonSecurity} and
\ref{thm:rancascadeonSecurity} imply the following three
discoveries:

\begin{enumerate}
\item [1)] Structures are essential to the security of networks against
cascading failure models of attacks,

\item [2)] There is a tradeoff between the role of structures and the
role of thresholds in security of networks, and

\item [3)] Neither power law distribution \cite{Bar1999} nor small world property \cite{WS1998} is
an obstacle of security of networks.

\end{enumerate}

The first discovery is a mathematical principle. From the viewpoint
of mathematics, we believe that structures determine the properties.
In so doing, a structural theory of networks would provide provable
guarantee for some of the key applications of network science. The
nature of networks are the networks themselves, instead of just
statistical measures of the networks. The investigation of
interactions and structures of interactions of networks is hence
essential to network theory and applications.

The second discovery explores that security of networks can be
achieved theoretically by structures of networks, and that there is
a tradeoff between the role of structures and the role of lifting of
the thresholds. This discovery is in sharp contrast to the current
practice of network security engineering which basically lifts the
thresholds. Exploring the tradeoffs between the role of structures
and the role of thresholds in security of networks would provide a
foundation for network security engineering, and hence it would be
exactly the subject of security theory of networks. Our discovery
here plays such a role.

The third discovery is also highly nontrivial. The reasons are:
intuitively speaking, power law allows us to attack a small number
of top degree nodes to generate a global cascading failure, and the
small world property means that spreading is so easy and so quick,
so that a small number of attacks may easily generate a global
cascading failure. This intuition is reasonable in some sense. In
fact, by observing our proofs, we know that there is only a small
window for us to construct networks to be both secure and to have
the power law and small world property.

Our discoveries imply that structure is a new, essential and
guaranteed source for security, and that the tradeoff between the
role of structure and the role of thresholds may provide both a full
understanding of security and new technology for security
engineering.

The proofs of Theorems~\ref{thm:cascadeonSecurity} and
\ref{thm:rancascadeonSecurity} provide a general framework for
theoretical analysis of security of networks. The main steps for
each of the uniform threshold security theorem and the random
threshold security theorem form the framework.

\bigskip

{\bf General framework}:

\begin{enumerate}
\item Small community phenomenon

The network is rich in quality communities of small sizes

\item The communities satisfy some more properties such as:
\begin{enumerate}

\item Each community has a few nodes dominating both internal and
external links

For each community $C$, let ${\rm dom}(C)$ be the dominating set of
$C$, which contains the hubs of the community $C$.

\item For each community $C$, the neighbors of nodes in $C$ outside
of $C$ are evenly distributed in different communities.

\end{enumerate}

\item We say that a community is strong, if it will never be
infected by the collection of outside communities, unless it has
already been infected by nodes in the community itself.

There are a huge number of strong communities.

\item By modulo the small communities, we can extract an
infection priority tree of the network.

\item Infections among the strong communities must be triggered by
an edge in the infection priority tree of the network.

\item The infection priority tree of the network has height $O(\log n)$.

\end{enumerate}

(1) provides a foundation for community analysis of the security of
networks. (2) ensures that there is a huge number of strong
communities. The existence of the infection priority tree $T$ is the
key to our proofs of the security theorems. (4), (5) and (6) ensure
that cascading procedure among the strong communities has a path of
length $O(\log n)$.

The general framework above provides not only a methodology to
theoretically analyze the security of networks, but also new
technology for enhancing security of real world networks.

\section{Threshold Theorem of Robustness of PA } \label{sec:cascade in
PA}

In this section, we prove Theorems \ref{thm:cascadeonPA} and
\ref{thm:negative}.

Suppose that $G=(V,E)$ is a network constructed from the PA model.

Given a node $v\in V$, we say that $v$ is {\it vulnerable} if one
infected neighbor is enough to infect it, or equivalently, its
degree $d_v$ is at most $1/\phi$.

The proof of Theorem \ref{thm:cascadeonPA} mainly consists of two
steps:

\begin{enumerate}
\item [(1)] By the definition of $G$, there is a large connected
component, $C$ say, in the subgraph induced by all the vulnerable
nodes in $G$.

In this case, if one node in $C$ is targeted, then all nodes in $C$
become infected.

\item [(2)] $G$ is an expander in the sense that the conductance of $G$
is large.

Therefore the set of infected nodes $C$ certainly infect new nodes
in $V\setminus C$ due to the reason that $\phi$ fraction neighbors
of $v$ are in $C$. This cascading procedure will continue until the
whole $G$ or a large part of $G$ being infected.

The second step of our proof implies that an expander-like graph is
unlikely to be robust.

\end{enumerate}

The proof of Theorem \ref{thm:negative} follows from a simple
observation that if $\phi$ is larger than $l/d$, then there is no
vulnerable node in $G$. For each node, if it is not targeted in the
initial random errors, then it cannot be infected unless at least
$l$ of its neighbors are infected. On the other hand, this is
unlikely to happen at the beginning when we randomly pick the
initial set of size $k=o(n^{\frac{l}{l+1}})$.

At first, we prove a basic version of Theorem \ref{thm:cascadeonPA}
for the case $\varepsilon=1$ in~\ref{subsec:weakversion}. The proof
of the main theorem will be developed by tightening the parameters
in Subsection \ref{subsec:strongversion}. In the end of this
section, we prove Theorem \ref{thm:negative}.

\subsection{Global Cascading Theorem of Single Node  }
\label{subsec:weakversion}

Before proving the full Theorem \ref{thm:cascadeonPA}, we prove a
basic result of the theorem for the case of $\varepsilon=1$.

\begin{theorem} \label{thm:weakversion}
There exists a positive integer $d_0$ such that almost surely (over
the construction of $G$), the following inequality holds:
\[\Pr\limits_{v\in_{\rm R} V}[{\rm inf}_{G}^\phi(\{v\})=V]\geq\frac{1}{2},\]
where $\phi=\frac{1}{2d}$.
\end{theorem}

\begin{proof}
We estimate the degree of each vertex. Denote by $v_i$ the $i$-the
vertex in $G$.

By Lemma~\ref{lem:degreeexpectation}, for time step $\frac{n}{4}$,
the expected degree of each vertex created after time step
$\frac{n}{4}$ is at most $2d$. So the expected number of nodes whose
degrees are at most $2d$ is $\frac{3}{4}n$, which correspond to last
$\frac{3}{4}n$ nodes.

From now on, we assume that there are $\frac{3}{4}n$ nodes (not
necessarily the last ones) which have degree at most $2d$ in $G$.

(In fact, a small deficit around $\frac{3}{4}n$, for instance,
$(\frac{3}{4}-\epsilon)n$ for some small $\epsilon$, does not
influence our analysis at all.  This will happen almost surely.)

Let $W$ be a set of all nodes having degrees at most $2d$.
 Note that $W$ is exactly the
set of all vulnerable nodes in $G$. Let $G_W$ be the induced
subgraph of $W$ in $G$.

We will show that with probability $1-o(1)$, the largest connected
component in $G_W$ has size at least $n/2$.

To explain our ideas without being trapped by complicated
parameters, we first prove a weak version of the conclusion.

\begin{lemma} \label{lem:weakccsize}
The size of the largest connected component in $G_W$ is almost
surely at least $\frac{n}{4}$.
\end{lemma}

\begin{proof}
Suppose to the contrary that the lemma fails to hold. We will show
that with probability $1-o(1)$, the number of connected components
of $G_W$ is $1$. In this case, the size of the largest connected
component is almost surely larger than $\frac{n}{4}$, contradicting
the assumption.

In this proof below, a connected component means a connected
component of $G_W$.

Suppose that $\{v_1,v_2,\ldots,v_{\frac{3}{4}n}\}$ is the set $W$
listed by the natural time ordering of nodes to be created. For
$j=1,2,\cdots, \frac{3}{4}n$, let $t_j$ be the time step at which
$v_j$ is created.

Let $m_1=\frac{5}{16}n$. Let $W_1=\{v_1, v_2,\cdots, v_{m_1}\}$, and
$W_2=W\setminus W_1$. We use $G_{W_1}[t]$ to denote the graph
induced by $W_1$ at the end of time step $t$.

Suppose for the worst case, that $G_{W_1}[t_{{m_1}}]$ is an
independent set, i.e., there is no even one edge among nodes in
$W_1$ at the end of time step $t_{m_1}$.

For every integer $i\in[1,\frac{7}{16}n]$, consider the influence of
node $v_{m_1+i}$ on the number of connected components in the
current graph. Let $\tau_i$ be the probability that there is an edge
from $v_{m_1+i}$ to some node in $\{v_1,\ldots,v_{m_1+i-1}\}$.

By the construction of $G$, the volume of
$\{v_1,\ldots,v_{m_1+i-1}\}$ is at least $(\frac{5}{16} n +i-1)d $,
and the volume of the graph constructed at the end of time step
$t_{m_1+i}$ is at most $2(\frac{n}{4}+m_1+i)d=2(\frac{9}{16}n+i)d$.

Thus
$$\tau_i\geq \frac{(\frac{5}{16}n+i-1)d}{2(\frac{9}{16}n+i)d}>\frac{1}{4}.$$

Let $N_1$ be the current volume of the largest connected component,
and $N_2$ be the current volume of all the nodes in $W$. Then
$N_1\leq\frac{n}{4}\cdot 2d=\frac{nd}{2}$, $N_2\geq
N_1+(\frac{5}{16}n-\frac{n}{4}+i-1)d\geq N_1+\frac{nd}{16}$. Let
$\rho=\frac{N_1}{N_2}$ be the probability that an edge of
$v_{m_1+i}$ connecting to the current largest connected component.
$\rho$ is also the upper bound of the probability that an edge of
$v_{m_1+i}$ connecting to some predetermined connected component. So
$\rho\leq \frac{8}{9}$. Let $\Delta_i$ be the difference of the
numbers of connected components  of the graphs after and before the
appearance of $v_{m_1+i}$ and its $d$ edges. A positive $\Delta_i$
means this number increases and otherwise decreases. Let $p$ be the
probability of $\Delta_i<0$, $p_0$ be the probability of
$\Delta_i=0$ and $p_1$ be the probability of $\Delta_i>0$. Note that
$\Delta_i>0$ means all the $d$ edges of $v_{m_1+i}$ do not connect
to any node in current $W$, and then we have
$$p_1\leq(1-\tau_i)^d.$$
$\Delta_i=0$ means that there are $j$ ($1\leq j\leq d$) edges of
$v_{m_1+i}$ join a single connected component while others do not.
We have
\begin{eqnarray*}
p_0 &\leq& \sum\limits_{j=1}^d \left(\begin{array}{c} d \\ j
\end{array}\right) \tau_i^j(1-\tau_i)^{d-j}\rho^{j-1}\\
&\leq& \sum\limits_{j=1}^d \left(\begin{array}{c} d \\ j
\end{array}\right) \tau_i^j(1-\tau_i)^{d-j}(\frac{8}{9})^{j-1}\\
&=& \frac{9}{8}\left[\left[(1-\tau_i)+\frac{8}{9}\tau_i\right]^d-(1-\tau_i)^d\right]\\
&=& \frac{9}{8}\left[(1-\frac{1}{9}\tau_i)^d-(1-\tau_i)^d\right].
\end{eqnarray*}
Since $p+p_0+p_1=1$, the expectation of $\Delta_i$ satisfies
$$E(\Delta_i)\leq p_1-p=2p_1+p_0-1\leq \frac{9}{8}(1-\frac{1}{9}\tau_i)^d+\frac{7}{8}(1-\tau_i)^d-1.$$
Since $\tau_i>\frac{1}{4}$, there must be some constant $d'$ such
that for any integer $d\geq d'$, $E(\Delta_i)\leq -\frac{5}{6}$. On
this condition, the number of connected components in the end is
expected to be a negative number. To prove that the number reduces
to the minimum possible number $1$, we use the supermartingale
inequality (see \cite{CL06}, Theorem 2.32). Let $m_2=\frac{7}{16}n$.
We consider the totally reduced amount compared with the initial
number $m_1$ at step $i$ of the last $m_2$ steps as a random
variable $X_i$ ($0\leq i\leq m_2$). We compute the totally reduced
amount by summing up all the reduced numbers in the former steps.
Keep it in mind that all the discussion is under the assumption that
there is no connected component of size at least $\frac{n}{4}$ in
the end. So the totally reduced amount would exceed $m_1$. By
definition, $X_0=0$. At each step, the reduced number is expected to
be at least $5/6$. Let $Y_i=X_i-\frac{5}{6}i$. Then $Y_0=0$,
$E[Y_i|Y_{i-1}]\geq Y_{i-1}$, and so $Y_0,Y_1,\ldots,Y_{m_2}$ is a
supermartingale. To show that with high probability, $X_{m_2}$ is at
least $m_1-1$, we only have to show that with probability $o(1)$,
$Y_{m_2}$ is no more than
$m_1-2-\frac{5}{6}\cdot\frac{7}{16}n=-(\frac{5}{96}n+2)$. By the
definition of the PA model, we know that $Y_i-E[Y_i|Y_{i-1}]\leq d$
and ${\rm Var}[Y_i|Y_{i-1}]=E[(Y_i-E[Y_i|Y_{i-1}])^2|Y_{i-1}]\leq
d^2$. Thus
$$\Pr\left[Y_{m_2}\leq-(\frac{5}{96}n+2)\right]\leq \exp\left(-\frac{(5n/96+2)^2}{2(d^2 m_2+d (5n/96+2)/3)}\right)=\exp(-\Omega(n)).$$
This means that under our assumption, with probability
$1-\exp(-\Omega(n))$, the number of connected components in the end
reduces to $1$. Lemma~\ref{lem:weakccsize} follows.
\end{proof}

The idea of the proof of Lemma \ref{lem:weakccsize} is to assume
that there is no connection for the part of the first coming nodes
(the first $\frac{5}{16}n$ nodes, slightly larger than
$\frac{n}{4}$), and then show that the remaining $\frac{7}{16}n$
nodes (slightly larger than $\frac{5}{16}n$) combine together to
form a large connected component. By the proof, it is also valid to
choose $m_1=\frac{n}{8}+\delta n$ (slightly larger than
$\frac{n}{8}$) and $m_2=\frac{n}{8}+2\delta n$ (slightly larger than
$m_1$), where $\delta$ is a small constant. Then by a similar
argument, we can show that there must be some constant $d_0'$
(relating to $\delta$) such that for any integer $d\geq d_0'$, when
the first $m_1+m_2=\frac{n}{4}+3\delta n$ nodes come, with
probability $1-o(1)$, there is a connected component of size at
least $\frac{n}{8}$. A key observation here is that, the current
number of connected components is at most $\frac{n}{4}+3\delta
n-\frac{n}{8}+1=\frac{n}{8}+3\delta n+1$. So by using the next
$\frac{n}{8}+4\delta n$ nodes, we can prove that there must be some
constant $d_1'$ (relating to $\delta$) such that for any integer
$d\geq d_1'$, when the first $\frac{3}{8}n+7\delta n$ nodes come,
with probability $1-o(1)$, there is a connected component of size at
least $\frac{n}{4}$. So recursively, using the following
$\frac{n}{8}+8\delta n$ nodes, we have that there must be some
constant $d_2'$ (relating to $\delta$), such that for any integer
$d\geq d_2'$, when the first $\frac{n}{2}n+15\delta n$ nodes come,
with probability $1-o(1)$, there is a connected component of size at
least $\frac{3}{8}n$. At last, using the remaining
$\frac{n}{8}+16\delta n$ nodes, we have that there must be some
constant $d_3'$ (relating to $\delta$), such that for any integer
$d\geq d_3'$, when all the nodes in $W$ come, with probability
$1-o(1)$, there is a connected component of size at least
$\frac{n}{2}$. Choosing
$\delta=\frac{1}{8(1+2+4+8+16)}=\frac{1}{248}$ makes the above
analysis work. Setting $d'=\max\{d_0',d_1',d_2',d_3'\}$, we have the
following lemma.

\begin{lemma} \label{lem:strongccsize}
There is a constant $d'$ such that for any $d\geq d'$, with
probability $1-o(1)$, the size of the largest connected component in
$G_W$ is at least $\frac{n}{2}$.
\end{lemma}

Denote by $S$ the largest connected component in $G_W$. When we
randomly and uniformly choose an initially infected node in $G$,
once it falls in $S$ or its neighbors, then the whole $S$ (including
at least $n/2$ vulnerable nodes) will be infected. This event
happens with probability at least $1/2$. So next, we only have to
show that based on the infected $S$, the cascading procedure will
sweep over the whole graph $G$, which completes the proof of Theorem
\ref{thm:weakversion}.

Mihail, Papadimitriou and Saberi \cite{MPS03} have shown that the
graph constructed by the PA model almost surely has a constant
conductance depending on $d$. Formally, they showed that when $d\geq
2$, for any positive constant $c<2(d-1)-1$, there exists an
$\alpha=\min\{\frac{d-1}{2}-\frac{c+1}{4},\frac{1}{5},\frac{(d-1)\ln2-(2\ln5)/5}{2(\ln
d+\ln2+1)}\}$ such that $\Pr[\Phi
(G)\leq\frac{\alpha}{\alpha+d}]=o(n^{-c})$ (see \cite{MPS03},
Theorem 1), where we use $\Phi (G)$ to denote the conductance of
$G$. By their proof, this result can be easily modified to the
following lemma.

\begin{lemma} \label{lem:condforPA}
There exists a constant $d''$ such that for any integer $d\geq d''$
and any $\alpha<\min\{\frac{d-1}{2}-\frac{1}{4},\sqrt{d}\}$, we have
$$\Pr\limits_{G\in_R \mathcal{P}(n,d)}[\Phi (G)\leq\frac{\alpha}{d+\alpha}]=o(1).$$
\end{lemma}

\begin{proof}
The proof follows from that of \cite{MPS03} except for different
choices of the parameters. We introduce the ideas here, and refer to
\cite{MPS03} for details. Let $\phi(G)$ be the edge expansion of
graph $G$ which is defined as $\phi(G)=\min_{S\subseteq V}
\frac{E(S,\overline{S})}{\min\{|S|,|\overline{S}|\}}$. In the PA
model, $\phi(G)$ can be used to bound the conductance,
$\Phi(G)\geq\frac{\phi(G)}{d+\phi(G)}$. So we only have to prove
$\Pr[\phi(G)\leq\alpha]=o(1)$. Let $2\leq k\leq n/2$. Consider all
the subset of nodes of size at most $k$, and then we can conclude
that
$$\Pr[\phi(G)\leq\alpha]\leq\sum_{k=2}^{n/2} \alpha k
(\frac{ed}{\alpha})^{2\alpha k} (\frac{k}{n})^{(d-1-2\alpha)k}.$$
For the $O(n)$ terms in this summation, if we upper bound the
leading term by $o(n^{-1})$, then the sum is upper bounded by
$o(1)$. We study the function $f(k)=\alpha k
(\frac{ed}{\alpha})^{2\alpha k} (\frac{k}{n})^{(d-1-2\alpha)k}$ for
$2\leq k\leq n/2$. It can be shown that there is a real number $x$
in the interval $[2,n/2]$ such that $f(k)$ monotonically decreases
in $[2,x]$ and monotonically increases in $[x,n/2]$. Thus the
leading term is either $f(2)$ or $f(n/2)$. If
$\alpha<\frac{d-1}{2}-\frac{1}{4}$,
$f(2)=2\alpha(\frac{ed}{\alpha})^{4\alpha}(\frac{2}{n})^{2(d-1-2\alpha)}=o(n^{-1})$.
On the other hand, since $f(n/2)=\frac{\alpha
n}{2}[(\frac{ed}{\alpha})^{2\alpha}(\frac{1}{2})^{d-1-2\alpha}]^{n/2}$,
there must be some constant $d''$ such that for any integer $d\geq
d''$, if $\alpha<\sqrt{d}$, then the product in the square bracket
is less than $1$. So $f(n/2)$ decreases exponentially as $n$
increases. This completes the proof of Lemma \ref{lem:condforPA}.
\end{proof}

Lemma \ref{lem:condforPA} guarantees that with probability $1-o(1)$,
the conductance of $G$ is at least $\frac{1}{2\sqrt{d}}$. On this
condition, we show that the cascading starting from $S$ will spread
all over the whole graph. Since every node in $G$ has degree at
least $d$ and ${\rm vol}(S)$ is $2nd$, we have ${\rm
vol}(S)\geq\frac{1}{2}nd$ and so ${\rm vol}(\overline{S})=2nd-{\rm
vol}(S)\leq\frac{3}{2}nd$. If ${\rm vol}(S)\leq nd$, that is, ${\rm
vol}(S)$ is no more than half of ${\rm vol}(G)$, then
$E(\overline{S},S)\geq {\rm vol}(S)\cdot\Phi
(S)\geq\frac{n\sqrt{d}}{4}$, and $\frac{E(\overline{S},S)}{{\rm
vol}(\overline{S})}\geq\frac{1}{6\sqrt{d}}$. For each node in
$\overline{S}$, let $E_S(v)$ be the number of nodes in $S$ that are
incident to some node in $S$ and $d_v$ be the degree of $v$. Then we
have
$$\frac{E(\overline{S},S)}{{\rm
vol}(\overline{S})}=\frac{\sum\limits_{v\in \overline{S}}
E_S(v)}{\sum\limits_{v\in \overline{S}} d_v}
\geq\frac{1}{6\sqrt{d}}.$$

By averaging, there must be some node $v\in\overline{S}$ whose at
least $\frac{1}{6\sqrt{d}}$ fraction of neighbors are infected. When
$d\geq 9$, this fraction is at least $\phi=\frac{1}{2d}$, and $v$ is
also infected. Add $v$ into $S$ and continue until ${\rm vol}(S)\geq
nd$ and ${\rm vol}(\overline{S})\leq nd$. Now
$\frac{E(\overline{S},S)}{{\rm vol}(\overline{S})}\geq \Phi
(G)\geq\frac{1}{2\sqrt{d}}$, which is larger than $\phi$. Thus by
averaging again, we know that there is a node $v\in\overline{S}$
being infected. Recursively, the whole graph $G$ will be infected.
The proof of Theorem \ref{thm:weakversion} is completed by choosing
$d_0=\max\{d',d'',9\}$.
\end{proof}

\subsection{Global Cascading Theorem of PA}
\label{subsec:strongversion}

In this subsection, we prove Theorem \ref{thm:cascadeonPA}.

\begin{proof} (Proof of Theorem \ref{thm:cascadeonPA})
The proof of Theorem \ref{thm:cascadeonPA} follows the proof of
Theorem \ref{thm:weakversion} step by step with tighter parameters.
By Lemma \ref{lem:degreeexpectation}, we suppose that there are
$(1-\frac{1}{(1+\varepsilon)^2})n$ nodes having degree at most
$(1+\varepsilon)d$. Denote by $W$ the set of them, and $W$ are
exactly the set of vulnerable nodes. Let
$p=\frac{2}{3}(1-\frac{1}{(1+\varepsilon)^2})$. By the proof of
Lemma \ref{lem:strongccsize}, we know that there exists an integer
$d'$ which only relates to $\varepsilon$ such that for any integer
$d\geq d'$, with probability $1-o(1)$, there exists a connected
component $S$ of size at least $pn$ in $G_W$. If we uniformly pick a
random initial node, then with probability at least $p$, it falls in
$S$ or its neighbors, which makes the whole $S$ infected. Then we
only have to show that the infection based on $S$ will spread all
over the whole graph $G$.

Note that the volume of $S$ is at least $pnd$. So we can choose a
$d''$ (only relating to $\varepsilon$), such that for any integer
$d\geq d''$, by averaging, there exists a node $v\in\overline{S}$
whose at least $\phi$ fraction of neighbors are infected. Then add
$v$ to $S$ and continue the procedure until the volume of $S$
exceeds $nd$. By Lemma \ref{lem:condforPA}, the current conductance
of $\overline{S}$ is at least $\frac{1}{\sqrt{d}+1}$, also larger
than $\phi$. Thus by averaging again, there is a node
$v\in\overline{S}$ being infected. Recursively, the whole graph $G$
will be infected. We choose $d_0=\max\{d',d''\}$, and complete the
proof of Theorem \ref{thm:cascadeonPA}.
\end{proof}

\subsection{Robustness Theorem of Graphs} \label{subsec:failurecondition}

In this subsection, we prove Theorem \ref{thm:negative} which holds
for all simple graphs.

\begin{proof} (Proof of Theorem \ref{thm:negative})
First, we bound by $o(n^{-1})$ the probability that a single node is
infected by the initial set $S$ of size $k=o(n^{\frac{l}{l+1}})$.
Then Theorem \ref{thm:negative} follows immediately by the union
bound.

For a node $v\in V\setminus S$, denote by $D$ the degree of $v$.
Then $D\geq d$. We can suppose that $D=O(k)$, because otherwise,
since $\phi$ is a constant, $v$ cannot be infected even if all the
nodes in $S$ are neighbors of $v$. Let $t=\lfloor \frac{l}{d}\cdot
D+1\rfloor$. Since $G$ is a simple graph, on the condition that
$v\in V\setminus S$, $v$ is infected by $S$ with probability

$
 \sum\limits_{i=t}^D \frac{\left(\begin{array}{c} D \\ i
\end{array}\right) \cdot \left(\begin{array}{c}
n-1-D \\ k-i
\end{array}\right)}{\left(\begin{array}{c}
n-1 \\ k
\end{array}\right)} \\
= \sum\limits_{i=t}^D \left(\begin{array}{c} D \\ i
\end{array}\right) \cdot \frac{(n-1-D)!}{(k-i)!(n-1-k-(D-i))!}
\cdot \frac{k!(n-1-k)!}{(n-1)!}\\
\leq \sum\limits_{i=t}^D \left(\frac{De}{i}\right)^i \cdot \left(\frac{k}{n-D}\right)^i\\
\leq \sum\limits_{i=t}^D \left(\frac{dek}{l(n-D)}\right)^i, $

\noindent where $e=2.718\cdots$ is the natural logarithm. The first
``$\leq$" comes from the
inequality $\left(\begin{array}{c} D \\
i
\end{array}\right)\leq\left(\frac{De}{i}\right)^i$ and the second ``$\leq$" comes from
$i\geq\frac{Dl}{d}$. Note that for each term $i$, $i\geq l+1$. Since
$k=o(n^{\frac{l}{l+1}})$, this sum is at most $o(n^{-1})$. This
completes the proof of Theorem \ref{thm:negative}.
\end{proof}

\section{Overlapping Communities Undermine Security of
Networks}\label{sec:over}

As we have seen that the topological, probabilistic and
combinatorial properties in Theorems \ref{thm:Securityproperties},
\ref{thm:length}, and \ref{thm:injury} guarantee the security
theorems. In these proofs, the following properties are essential:

\begin{enumerate}
\item [(i)] The small community phenomenon.
\item [(ii)] Local heterogeneity

That is, the seed node of a community plays a central role in both
internal and external links of the community.
\item [(iii)] Randomness and uniformity among the global edges.
\item [(iv)] The existence of the infection priority tree of height
$O(\log n)$.

\end{enumerate}

We further analyze the corresponding roles of properties (i) - (iv)
above. The role of (i) is clear, since otherwise, it would be
possible a single targeted node in a large community may infect the
whole community which is large. (ii) and (iii) ensure that almost
all communities are strong. (iv) ensures that cascading among strong
communities has a path of short length.

Except for (i) - (iv) above, we notice that the small communities in
the security model are disjoint. Therefore there is no overlapping
community phenomenon in networks of the security model.

For nontrivial networks constructed from models such as the ER and
the PA models, we know that there is no even a community structure
in the networks. However overlapping communities seem universal in
real networks. Intuitively, overlapping communities undermine
security of the networks. The reason is that if a node, $v$ say, has
two communities, $C_1$ and $C_2$ say, then attack on $v$ is in fact
attacks on both the communities $C_1$ and $C_2$. We show that this
intuition is correct.

To verify the conclusion, we modify the security model as follows.

\begin{definition}\label{def:security-over} (Overlapping model) Given homophyly exponent
$a$, $d_1\geq 2$, $d_2\geq 2$ and $d=d_1+d_2$. We construct a
network as follows.

\begin{enumerate}
\item [(1)] Let $G_2$ be an initial graph such that each node of
$G_2$ is called a seed node, and is associated with a distinct
color.

For $i>2$, suppose that $G_{i-1}$ has been defined., and let
$p_i=1/(\log i)^a$. We define $G_i$ as follows.

\item [(2)] Create a new node, $v$.

\item [(3)] With probability $p_i$, $v$ chooses a new color, $c$  say, in
which case:

\begin{enumerate}
\item We say that $v$ is a seed node,
\item Create an edge $(v,u)$, where $u$ is chosen with probability
proportional to the degrees of nodes in $G_{i-1}$,
\item Create $d_1-1$ edges $(v,u_j)$ for $j=1,2,\cdots, d_1-1$,
where each $u_j$ is chosen randomly and uniformly among all seed
nodes in $G_{i-1}$, and

\item Choose randomly and uniformly an old color, $c'$ say,
\item We say that $v$ has two colors, both $c$ and $c'$, and

\item Create $d_2$ edges $(v,w_k)$ for $k=1,2,\cdots, d_2$, where
each $w_k$ is chosen with probability proportional to the degrees of
all nodes sharing the old color $c'$ in $G_{i-1}$.

\end{enumerate}

\item [(4)] Otherwise, then $v$ chooses an old color, in which case,
then

\begin{enumerate}

\item Let $c$ be an old color chosen randomly and uniformly among
all colors appeared in $G_{i-1}$,

\item Let $c$ be the color of $v$, and
 \item Create $d$ edges $(v,x_l)$ for $l=1,2,\cdots, d$, where each $x_l$
is chosen with probability proportional to the degrees among all
nodes sharing color $c$.

\end{enumerate}

\end{enumerate}

\end{definition}

Suppose that $G$ is a network constructed from
Definition~\ref{def:security-over}. By definition, it is easy to see
that $G$ has the small diameter property. In Figure
\ref{fig:degree_distribution_N=10000_a=1.5_d=10}, we compare the
degree distributions of networks of the security model and the
overlapping model. The networks have $n=10,000$ nodes, homophyly
exponent $a=1.5$, $d=10$ and $d_1=d_2=5$ for the overlapping model.
The experiment shows that both the networks follow the same power
law.

\begin{figure}[htbp]
                 \centering
                 \includegraphics[width=3.2in]{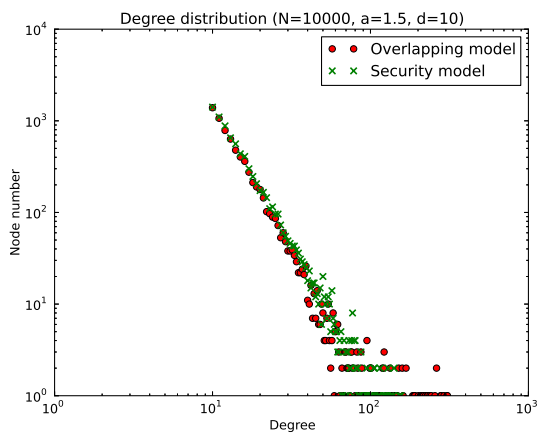}
                 \caption{Power law of networks of the security model and the overlapping model}
                 \label{fig:degree_distribution_N=10000_a=1.5_d=10}

\end{figure}

Let $G$ be a network constructed from the overlapping model. We
define a community to be the induced subgraph of a homochromatic
set. Clearly a community is connected. In Figure
\ref{fig:Conductance_N=10000_a=1.5_d=10}, we compare the
distribution of conductances of a network of the security model and
a network of the overlapping model with the same parameters as
above.

\begin{figure}[htbp]
                 \centering
                 \includegraphics[width=3.2in]{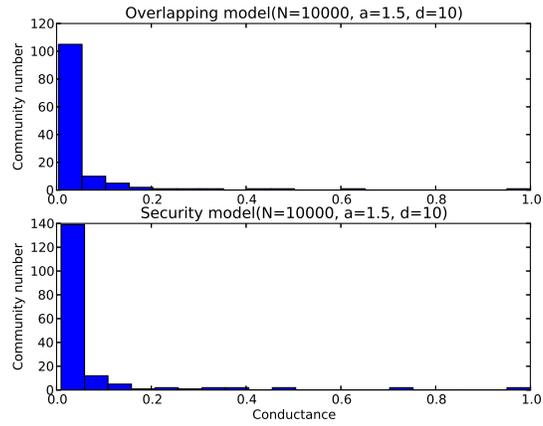}
                 \caption{Distributions of conductances of communities for networks of both the security model and the overlapping model}
                 \label{fig:Conductance_N=10000_a=1.5_d=10}

\end{figure}

From Figure \ref{fig:Conductance_N=10000_a=1.5_d=10}, we know that
distributions of conductances of all the communities are similar to
each other, and almost all are small. This shows that networks
constructed from the overlapping model are rich in small communities
too.

The only difference between $G$ and networks constructed from the
security model is that for each seed node of $G$, $v$ say, $v$
contains in $2$ communities. Our intuition is that overlapping
communities undermine security of the networks.

In Figure~\ref{fig:Security_plot_N=10000_d=10_a=1.5}, we compare the
security of networks constructed from both the security model and
the overlapping model for $n=10,000$, $a=1.5$, $d=10$, and
$d_1=d_2=5$.

\begin{figure}[htbp]
                 \centering
                 \includegraphics[width=3.2in]{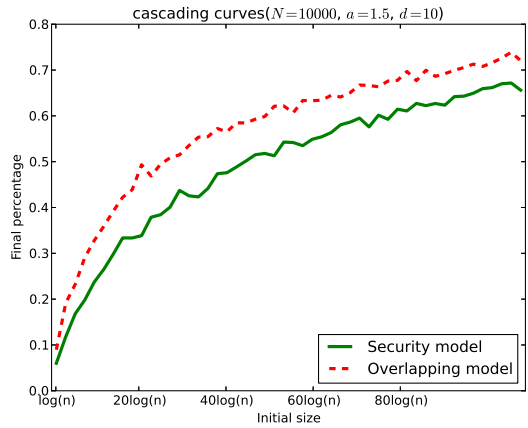}
                 \caption{Security curves}
                 \label{fig:Security_plot_N=10000_d=10_a=1.5}

\end{figure}

Experiments in Figure \ref{fig:Security_plot_N=10000_d=10_a=1.5}
show that the network constructed from the security model is more
secure than that of the overlapping model for attacks of all
small-scales. This verifies that overlapping communities do
undermine security of networks.

By this reason, we give up the phenomenon of overlapping communities
in our elementary security theory of networks.

However it is still an open issue to fully understand the
undermining of overlapping communities in security of networks.
Solving this problem may provide a new way to enhance security of
networks by distinguishing the different roles of a node in
different communities. It is not surprising we may need a way to
deal with the undermining effect of overlapping communities on
security of networks. In general, it is an interesting open question
to fully understand the roles of overlapping communities, since it
seems universal in many real networks. Sometimes, overlapping
communities are bad, for instance, every corrupt official confuses
his/her public and private roles.

\section{\bf Conclusions and future directions} \label{sec:conc}

 In this paper, we
proposed definitions of security and robustness of networks to
highlight the ability of complex networks to resist global cascading
failures caused by a small number of deliberate attacks and random
errors, respectively. We use the threshold cascading failure models
to simulate information spreading in networks.

We introduced a security model of networks such that networks
constructed from the model are provably secure under both uniform
and random threshold cascading failure models, and simultaneously
follow a power law, and satisfy the small world phenomenon with a
remarkable $O(\log n)$ time algorithm to find a short path between
arbitrarily given two nodes. This shows that networks constructed
from the security model are secure, follow the natural property of
power law, and allow a navigation of time complex $O(\log n)$.

The security model shows that dynamic and scale-free networks can be
secure for which homophyly, randomness and preferential attachment
are the underlying mechanisms, providing a principle for
investigating the security of networks theoretically and generally.

Our security theorems explore some new discoveries between the roles
of structures and of thresholds in the security of networks. The
proofs of the security theorems provide a general framework to
analyze both theoretically and practically security of networks.

It seems surprising that networks of the security model satisfy
simultaneously all the properties stated in the three theorems,
i.e., Theorems \ref{thm:Securityproperties}, \ref{thm:length} and
\ref{thm:injury}, and that a merging of the principles in Theorems
\ref{thm:Securityproperties}, \ref{thm:length} and \ref{thm:injury}
gives rise to the proofs of the security theorems, Theorems
\ref{thm:cascadeonSecurity} and \ref{thm:rancascadeonSecurity}. This
is a mathematical creation and mathematical beauty with immediate
and far-reaching implications in network communication and network
science.

On the other hand, the mechanisms of homophyly, randomness and
preferential attachment of the security model are natural selections
in evolutions of complex systems in both nature and society. This
may explain the reason why networks of the security model have the
remarkable properties here. This may also imply that the security
model reflects some of the natural laws and social principles. This
poses some fundamental questions such as: Does nature compute hard
problems? Does nature evolve safely? Does society organize securely
and stably? A possible approach to answering these questions could
be to explore the physical, biological and social science
understandings of the security model.

As usual, many real networks may not evolve as our security model.
This is not surprised. There are always some differences between
networks constructed from models and real networks. For instance, i)
nontrivial networks of the PA model fail to have a community
structure, but almost all real networks have, ii) nontrivial
networks of the ER model fail to have a community structure or power
law distribution, but almost all real networks have, and iii)
networks of the small world model fail to have a power law, but
almost all real networks have.

However, in our case, if real networks evolve in a way far from our
security model, then it may imply that the real networks are highly
unlikely to be secure, or worse, not even to be robust against a few
random errors. This situation means that we do really and urgently
need a theory to guarantee security of the networks in which we are
living.

The mechanisms of the security model are natural selections in
organizations of networks in both nature and society. However, the
construction of networks in Definition \ref{def:Securitymodel} is
carefully organized. A reader may wonder whether or not there is a
cheaper construction of the networks with less ingredients than that
in our definition. This could be possible, however, by our
understanding, security of networks cannot be achieved freely,
either in theory or in engineering.

A reader may wonder Definition \ref{def:Securitymodel} is a simple
modification of the PA model, the two models should give similar
networks. Why are the networks so different? It is true. However,
there are two more new ideas introduced in the security model: the
first is that every node has its own characteristic at the very
beginning of its birth, that is, either remarkable (with new color)
or normal (with old color), and the second is that two more natural
mechanisms are introduced to remarkable nodes and normal nodes
respectively. More importantly, the new ideas and the new mechanisms
introduce ordering and combinatoric principles in the construction
of networks. This perhaps explains that combinatorics plays a
remarkable role in networks, and that purely probabilistic and
single mechanism fails to capture complexity in nature and society.

By Theorem \ref{thm:cascadeonSecurity}, for $\phi=O(\frac{1}{\log^b
n})\ll 1/d$, networks generated from the security model is
$\phi$-robust and $\phi$-secure. For the same constant $d$, the
networks generated from both the PA model and the security model
have the same average degree. By Theorem
\ref{thm:cascadeonSecurity}, the security threshold for the security
model can be arbitrarily small as $n$ increases, while by Theorems
\ref{cor:PApositive} and \ref{cor:PAnegative}, the robustness
threshold for the PA model can only be the constant $1/d$. These
theorems indicate that the structure of a network is key to the
robustness and security although power law and small world
properties exist in both models. Neither of these two properties is
an obstacle to network robustness and security, while the small
community phenomenon and connection patterns among communities play
an essential role. Consequently, the security model provides an
algorithm to construct dynamically networks which are secure against
any attacks of small sizes under both uniform and random threshold
cascading failure models, and which satisfy all the useful
properties of usual networks.

Our results start a theoretical approach to network security.
However there is a huge number of important issues open, for which
we list some of them:

\begin{enumerate}

\item The role of homophyly exponent

We notice that the homophyly exponent $a$ in
Theorems~\ref{thm:Securityproperties}, ~\ref{thm:cascadeonSecurity}
and ~\ref{thm:rancascadeonSecurity} is greater than $1, 4$ and $6$
respectively, showing some differences among the fundamental
theorem, the uniform threshold security theorem and the random
threshold security theorem. The assumptions of $a>4$ and $a>6$ are
essentially used in the proofs of Theorems
~\ref{thm:cascadeonSecurity} and ~\ref{thm:rancascadeonSecurity}
respectively. By Theorem \ref{thm:length}, it seems necessary for
$a$ to be large to make sure that almost all communities are strong.
However for large $a$, the sizes of communities are also large, so
that attack on a single node in a community may infect all nodes of
the large community. Of course, for theoretical results, we only
need to prove the theorems for all sufficiently large $n$, in which
case, large $a$ is not a problem. In practice, the sizes of networks
are limited, in which case, it is necessary to choose appropriate
$a$ to make a balance to achieve the best possible security.
Fortunately we have shown experimentally in~\cite{LZPL2013a} that
even for just $a>1$, for small $n$, networks of the security model
are much more secure than that of both the ER and PA models under
both random and uniform threshold cascading failure models. This
poses a question to theoretically study the security theorems for
just $a>1$, which will be more helpful for practical applications.
Answering this question is not going to be easy, which calls for new
analysis or new ideas.

In our proofs, the average number of edges $d$ is assumed to be a
constant. This assumption has no effect on theoretical results for
all sufficiently large $n$. However, for fixed number $n$, the value
$d$ plays a role. Usually the larger is $d$, the less secure is the
network. This is reasonable, because, the larger is $d$, the denser
is the network. However this problem is interested in only practice.

In practice, many real networks may not be secure for which there
are too many reasons. However, security may have only one reason.
Our principles and theorems here provide a chance for us to examine
the reasons why a given real network is insecure. Once we know the
reasons of insecurity of a network, we may have ways to secure the
network.

\item Security vs robustness

Our theorems show that the security model in
Definition~\ref{def:Securitymodel} is secure, and that the
preferential attachment model in~\cite{Bar1999} is non-robust. It
would be interesting to find a model of networks (dynamic, and with
power law and small diameter property etc) that it is robust, but
insecure. The significance of answering this question is to fully
understand the robustness and the security of networks.

\item Criterions for security

Our security model provides a principle for security. However, it is
open to define criterions for security of a given network. This
poses fundamental open questions such as: What are the theoretical
criterions to measure quantitatively the security of real networks?
What are the best possible algorithms to compute security indices of
real networks?

\item Enhancing security

A new fundamental question closely related to security applications
is to enhance security of networks. In practice, we are given a
network, $G$ say, and asked to make a minimal modification of $G$ to
generate a network, $H$ say, such that $H$ keeps all the useful
properties of $G$ and such that $H$ is much more secure than $G$.
Our security model suggests some strategies for enhancing the
security of networks. However, theoretical study of this issue is
completely open.

\item Influence of structures

 To consider the influence of a
structure, $G_S$ say (the induced subgraph of $S$), instead of just
the set $S$ of nodes.

\item Fully understand the roles of mechanisms, structures in the
security of networks

 \item Security and game

To introduce games in security strategies of networks.

\item Security and diffusion models

To consider a variety of diffusion models according to different
applications, for example, the independent cascading failure model
in the context of marketing~\cite{GLM01a},~\cite{GLM01b}.

\item Security of weighted networks

 To study the security of weighted and directed versions of
networks to better capture new phenomena in the globalizing economic
networks etc.

\item Robustness of networks

Our results in Theorems~ \ref{cor:PApositive} and
~\ref{cor:PAnegative} have implications in applications due to the
fact that most real networks heavily depend on the preferential
attachment scheme without guaranteeing the large threshold for all
the vertices. This would imply that most real networks may not be
even robust against random errors (or random attacks). And more
importantly, the gap between the robustness threshold and
non-robustness threshold is small, it could be very easy for a
network to be non-robust. This means that robustness of networks is
not an issue we can take for granted, and that global failure of
real networks could be simply caused by random errors, instead of
deliberate attacks. For this reason, robustness needs to be studied
separately.

\end{enumerate}

Finally we emphasize that our theory is to investigate the roles of
structures, and the tradeoff between the role of structures and the
role of thresholds in the security of networks. In engineering,
security could be achieved by lifting the thresholds for all nodes,
without considering the roles of structures of networks. A
relatively long term challenge is to build a bridge between theory
and engineering of security of networks. In practice, one more tough
issue could be to distinguish positive and negative contents in the
cascading procedure, which is already not purely a scientific
problem.

\bibliographystyle{plain}
\bibliography{acmsmall-security-2013-bibfile}

\end{document}